\documentclass{article}
\usepackage{arxiv}

\usepackage[utf8]{inputenc} 
\usepackage[T1]{fontenc}    
\usepackage{hyperref}       
\usepackage{url}            
\usepackage{booktabs}       
\usepackage{amsfonts}       
\usepackage{nicefrac}       
\usepackage{microtype}      
\usepackage{lipsum}		
\usepackage{graphicx}

\usepackage{amsmath,amssymb}
\usepackage{latexsym,graphicx}
\usepackage{amsbsy,amscd,amsthm}
\usepackage[mathscr]{eucal}

\newcommand{\bm}[1]{\mbox{\boldmath$#1$}}
\def\eps{\varepsilon}
\def\al{\alpha}
\def\be{\beta}
\def\ga{\gamma}
\def\O{\mathcal{O}}
\newcommand{\se}{\sqrt{\eps}}

\newcommand{\bmU}{\bm{U}}
\newcommand{\bmD}{\bm{D}}
\newcommand{\bmF}{\bm{F}}
\newcommand{\bmI}{\bm{I}}
\newcommand{\bmG}{\bm{G}}
\newcommand{\bmZ}{\bm{Z}}
\newcommand{\bmxero}{\bm{0}}

\newtheorem{remark}{Remark}
\newtheorem{theorem}{Theorem}
\numberwithin{defn}{section}
\numberwithin{rem}{section}
\newtheorem{lemma}[theorem]{Lemma}

\title{Traveling pulse solutions in a three-component FitzHugh--Nagumo model}

\author{ Takashi~Teramoto 
\thanks{TT is partially supported by KAKENHI Grant-in-Aid for Scientific Research 17K05355.} \\
School of Medicine, \\ Asahikawa Medical University, \\
Asahikawa, 078-8510, Japan \\
	\texttt{teramoto@asahikawa-med.ac.jp} \\
	\And
	Peter van~Heijster  
	\thanks{PvH was supported under the Australian Research Councils Discovery Project DP190102545.}\\
	School of Mathematical Sciences, \\
	Queensland University of Technology, \\
	Brisbane, QLD 4001, Australia
}

\hypersetup{
pdftitle={Traveling pulse solutions in a three-component FitzHugh--Nagumo model},
pdfauthor={Takashi~Teramoto, Peter van~Heijster},
pdfkeywords={reaction-diffusion equations,
singular limit, action functional,existence, stability,
saddle node bifurcation},
}

\begin{document}
\maketitle

\begin{abstract}
We use geometric singular perturbation techniques combined with an action functional approach to 
study traveling pulse solutions in a three-component FitzHugh--Nagumo model. First, we derive the profile 
of traveling $1$-pulse solutions with undetermined width and propagating speed. Next, we compute 
the associated action functional for this profile from which we derive the conditions for existence 
and a saddle-node bifurcation as the zeros of the action functional and its derivatives. 
We obtain the same conditions by using a different analytical approach that exploits 
the singular limit of the problem. We also apply this methodology of the action functional to the problem 
for traveling $2$-pulse solutions and derive the explicit conditions for existence and a saddle-node bifurcation. 
From these we deduce a necessary condition for the existence of traveling $2$-pulse solutions. 
We end this article with a discussion related to Hopf bifurcations near the saddle-node bifurcation. 
\end{abstract}


\section{Introduction}
The study of spatially localized patterns in multi-component reaction-diffusion systems has a long history, see for instance the surveys of experimental and numerical studies in various physical and chemical contexts \cite{Purwins3, Vanag}. The myriad of  experimental and numerical studies highlight the necessity to develop a theoretical study of the existence, stability, bifurcation and dynamics of localized solutions \cite{Knobloch,MJWard}. 
Front and pulse solutions in one spatial dimension, and spot solutions 
in higher dimensions, have been of particular interest for the theoretical studies 
in, for instance, the singular limit of the two-component FitzHugh--Nagumo model
and Gray-Scott type models \cite{Doelman1, Kolokolnikov, Nishiura1}. 

The focus of this paper is on traveling pulse solutions 
in a three-component FitzHugh--Nagumo model. This model was 
originally proposed as a phenomenological model for the gas-discharged systems 
studied by Purwins et al. \cite{Purwins1, Purwins2, Purwins3}, and reformulated for the mathematical analysis in the singular limit by Doelman et al. \cite{vH_EXIS}. 
The results derived in these references indicate that this three-component 
model has richer and more complicated solutions (when compared 
to the original two-component model). 
For instance, on unbounded domains stable traveling spot solutions in higher dimensions \cite{vH_SPOT_ST} and stationary $2$-pulse solutions \cite{vH_EXIS} only exist in the extended three-component  model. 
In this paper, we will show the existence of traveling $2$-pulse solutions. Such clustered and localized moving solutions are specific to the following three-component model. 

The singularly perturbed three-component FitzHugh--Nagumo model 
under consideration is 
\begin{align}
 \label{FHNd}
 \left\{
\begin{aligned}
    U_t & = \eps^2 U_{xx} + U - U^3 - \eps(\al V + \be W + \ga) \,,\\  
        \tau V_t & = V_{xx} + U - V \,, \\
        \theta W_t & = D^2 W_{xx} + U - W \,, 
\end{aligned}
\right.
\end{align}
where $0 < \eps\ll 1; D>0; (x,t) \in \mathbb{R} \times \mathbb{R}^+$ and 
the parameters $\al, \be, \ga, D$ are assumed to be strictly
$\mathcal{O}(1)$ with respect to $\eps$. 
The small parameter $\eps$ plays the role of a perturbation parameter and
the fast $U$-component is weakly coupled to the two slow $V$- and $W$-components.
The system is bistable with two stable trivial background states ${\cal O}(\eps)$-close 
to $(U,V,W) = \pm(1,1,1)$. The singular perturbed nature of the problem has enabled mathematicians to study 
the various aspects of localized solutions supported by (\ref{FHNd}) intensively \cite{Martina,vH_EXIS,Kajiwara, Nishiura5,vH_STAB,vH_FRONT,vH_SPOT,vH_SPOT_ST}. 
For instance, Doelman et al. \cite{vH_EXIS, vH_STAB} determined under what conditions on the 
system parameters the model supports stable stationary pulse solutions. The authors adopted geometric singular perturbation theory (GSPT) with a Melnkov-type integral and an Evans function approaches to explicitly derive the existence and stability conditions for stationary $1$-pulse and $2$-pulse solutions. 

In \cite{vH_ACTION1}, we reconsidered the same problem and developed a methodology based on the variational formulation of the problem. This methodology consists of two parts: construction of GSPT solutions with the undermined pulse width and 
computation of an action functional associated with the GSPT solution profile. 
The pulse width will be determined by the extrema of this action functional and stable solutions will be minimizers. 
The action functional was originally used in a series of papers by Chen et al. \cite{CH1,CH2,CH3,CH4} to study the two-component FitzHugh--Nagumo model away from its singular limit with the activator $U$ strongly coupled with inhibitor in the $U$-equations. By investigating the extrema in the variational structure, the authors proved  the existence and stability of front and pulse solutions. It is worth noting that 
recently this case has been rigorously studied as well by Chen and Choi \cite{CH6}, 
and numerical studies on the stable traveling pulse solutions were given by Choi 
and Connors \cite{CH7}. Chen et al also considered the weak coupling case and 
derived the explicit conditions for existence and uniqueness of traveling pulse solutions  in the two-component model \cite{CH5}. 

 In our previous studies \cite{vH_ACTION1,vH_ACTION2}, we assumed that 
 $\tau$ and $\theta$ were $\mathcal{O}(1)$ with respect to $\eps$, and 
 dealt with the existence and stability of only the stationary pulse solutions (since traveling pulse solutions necessarily have $\tau$ and/or $\theta$ 
 of $\mathcal{O}(1/\eps^2)$  \cite{vH_EXIS, vH_STAB}). 
 In this paper, we are going to extend the methodology to this parameter 
 regime where the time constants $\tau$ and $\theta$ are 
 set to $\mathcal{O}(1/\eps^2)$. In this setting, the stationary pulse solutions potentially bifurcate to traveling pulse solutions or 
 breather solutions \cite{vH_EXIS, vH_STAB} and, to complicate the analysis, the essential spectrum is asymptotically close to the origin and additional eigenvalues pop out of the essential spectrum \cite{Martina, Martina2, vH_STAB}. See also Remark~\ref{R:0}. Here, we focus on the existence and stability of traveling $1$-pulse and $2$-pulse solutions. Similar results for traveling pulse solutions in a two-component system were  given in \cite{CH2,CH5}. In those papers, the mono-stable case with a different asymptotic scaling was treated first \cite{CH2} and the authors later extended their analysis to the bistable system \cite{CH5}, see also Remark~\ref{R:1}. 

We introduce the atypical co-moving frame $z := c(x-\eps^2 ct)$, originally proposed in \cite{Heinze}, 
to study traveling $1$-pulse solutions $\bar{Z}_{p}$ and traveling $2$-pulse solutions $\bar{Z}_{2p}$.
See Fig.~\ref{fig01_NEW} for an example of a traveling $1$-pulse and $2$-pulse solution. That is, a traveling pulse solution $\bar{Z}_p$ or $\bar{Z}_{2p}$, 
 represented by $(\bar{U},\bar{V},\bar{W})(z)$ with a wave speed 
 $\eps^2 c$, solves 
\begin{eqnarray}
 \label{comoving}
 \left\{
\begin{array}{rl}
   - \eps^2 c^2 \bar{U}_z & = \eps^2 c^2 \bar{U}_{zz} + \bar{U} - \bar{U}^3 - \eps(\al \bar{V} + \be \bar{W} + \ga) \,,\\  
   -  c^2 \hat{\tau} \bar{V}_z  & = c^2 \bar{V}_{zz} + \bar{U} - \bar{V} \,,\\  
   -  c^2 \hat{\theta} \bar{W}_z  & = D^2 c^2 \bar{W}_{zz} + \bar{U} - \bar{W} \,,
\end{array}
\right.
\end{eqnarray}
where $(\hat{\tau}, \hat{\theta}) := (\eps^2 \tau,\eps^2\theta)$ such that $\hat{\tau}$ and $\hat{\theta}$ are now $\mathcal{O}(1)$.
The second and third linear equations for the $V$- and $W$-components 
satisfy $\bar{V} = \mathcal{L}_{1c} \bar{U}$ and $\bar{W} = \mathcal{L}_{2c} \bar{U}$, 
with the operators 
\begin{eqnarray} 
\label{OPER}
\mathcal{L}_{1c} := \left(- c^2 \frac{d^2}{d z^2} - \hat{\tau} c^2 \frac{d}{d z} + 1 \right)^{-1}\,, \quad 
\mathcal{L}_{2c} :=\left(- D^2 c^2 \frac{d^2}{d z^2} -\hat{\theta} c^2 \frac{d}{d z} + 1 \right)^{-1} .\end{eqnarray} 
In this article, we set the rescaled time constants $(\hat{\tau}, \hat{\theta})$ to $(1, D^2)$
to be able to apply the variational formulation with an action functional.  
 In particular, in this setting the operators $\mathcal{L}_{1c}$ and $ \mathcal{L}_{2c}$ become self-adjoint operators, 
 i.e., $\langle v, \mathcal{L}_{1,2c} w \rangle_{L_{ex}^2} = \langle \mathcal{L}_{1,2c} v, w \rangle_{L_{ex}^2}$ for any $v, w$ in the weighted Hilbert space $L_{ex}^2$, 
corresponding to the inner product $\displaystyle{ \langle v, w \rangle_{L_{ex}^2}
= \int_{\mathbb{R}} e^x\,v \,w \,dx  }$ \cite{CH2,Heinze}. 
\begin{figure}[ht!]
 \centering
 \includegraphics[width=12cm]{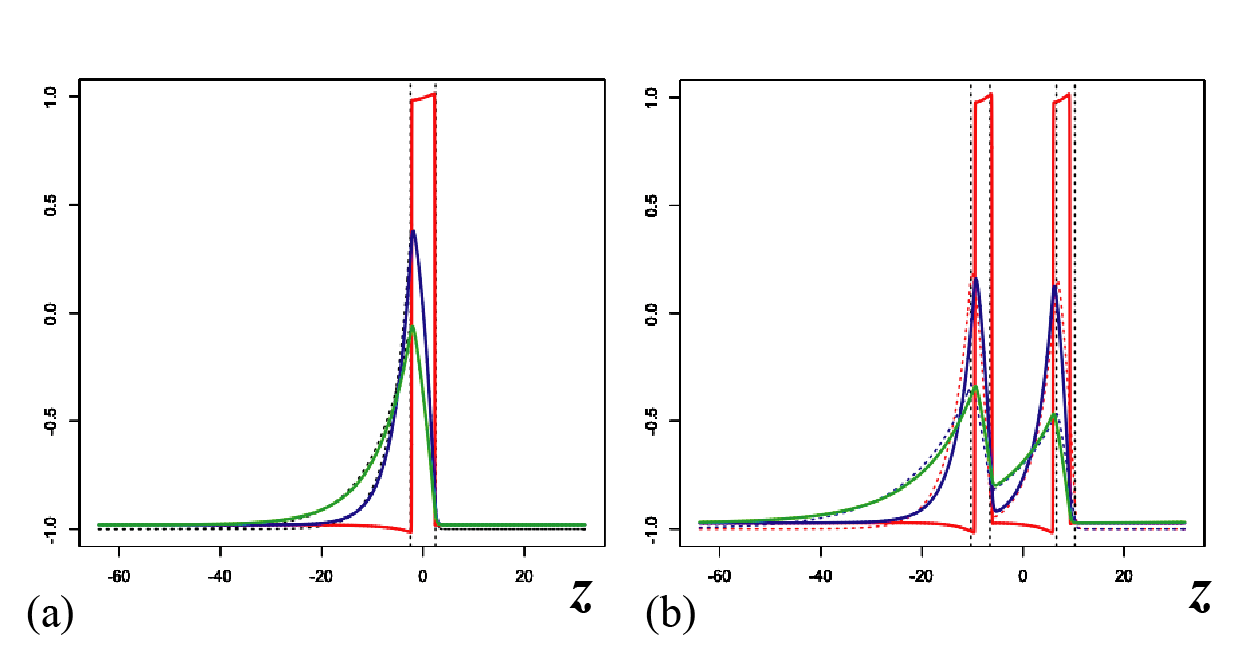}
  \caption{Right-going traveling $1$-pulse and $2$-pulse solutions. 
    Solid curves indicate the traveling pulse profiles obtained by numerical continuation of the original model (\ref{FHNd}), see Appendix~\ref{SS:NUM} for more details on the numerical continuation. Red is the $U$-profile, blue the  $V$-profile and green the $W$-profile.
Dotted curves indicate the leading order solution profiles 
as derived in this article, see Theorem~\ref{TH:1P} and Theorem~\ref{TH:2P}. 
    The system parameters used in (a) are $(\al, \be, \ga,  D^2, \eps) = (2, 1, 1, 2, 0.02)$ and the pulse width and propagating speed are $(2 z^*/c, \eps^2 c) = 
    (4.9592, 1.4175 \times 10^{-3})$ with $(z^*, c) = (8.78696, 3.54371)$.
     The system parameters used in (b) are $(\al, \be, \ga,  D^2, \eps) = (4, -1, 0.8, 3, 0.025)$ and the pulse characteristics and propagating speed are $(2y_1/c, 2y_2/c, 2y_3/c, \eps^2 c) = (3.7585, 13.1692, 3.6680, 2.3567 \times 10^{-3})$ with $(y_1, y_2, y_3, c) = (7.08604, 24.8289, 6.91551, 3.77071)$. Here, $2y_1/c$ and $2y_3/c$ are related to the width of the two pulses, while $2y_2/c$ is related to the distance between the two pulses.
}
  \label{fig01_NEW}
\end{figure}

The main results of this article related to the existence of traveling $1$-pulse and $2$-pulse solutions are stated in Theorem~\ref{TH:1P} and Theorem~\ref{TH:2P} below.
\begin{theorem}
\label{TH:1P}
Let $(\hat{\tau}, \hat{\theta}) = (\varepsilon^2 \tau, \varepsilon^2 \theta) = (1, D^2)$ and let $(\alpha, \beta, \gamma, D)$ be such that 
\begin{eqnarray}
\label{existence_criterion_one}
\left\{
\begin{array}{l}
\displaystyle{\alpha \bar{V}(z^*) +  \beta \bar{W}(z^*) + \gamma + \frac{\sqrt{2}}{3} c= 0} \,, \\
\displaystyle{\alpha \bar{V}(-z^*) +  \beta \bar{W}(-z^*) + \gamma - \frac{\sqrt{2}}{3}c = 0 }
\end{array}
\right.
\end{eqnarray}
has positive solutions $z^*$ and $c$. Then, 
for small $\varepsilon$ enough, (\ref{FHNd}) supports a traveling $1$-pulse solution $\bar{Z}_{p} = (\bar{U}, \bar{V}, \bar{W})$ that travels with propagating speed $\varepsilon^2 c$ and has leading order width $2z^*/c$, and which 
goes asymptotically to $(U_b,U_b,U_b)$ as $x \to \pm \infty$. Here, $U_b$ is the most negative root of the cubic 
equation $u^3 - u + \varepsilon ( (\alpha + \beta) u + \gamma) = 0$ and
the scaled values of the $V$-component at $ \pm z^*$ are, to leading order, given by
\begin{eqnarray} 
\label{VV}
\displaystyle{ \bar{V}(\pm z^*) =  \frac{1}{\phi_v} \left( \mp1 \pm (1\mp\phi_v) e^{\mp(1\pm\phi_v) z^*}  \right) }
\end{eqnarray}
where $\displaystyle{\phi_v = \sqrt{1 + \frac{4}{c^2}}}$. 
The values of the $W$-component at $ \pm z^*$ are as (\ref{VV}) but with $\phi_v$ replaced by $\displaystyle{\phi_w = \sqrt{1 + \frac{4}{c^2D^2}}}$. 

In addition, a saddle node bifurcation occurs on the solution branch of 
a traveling $1$-pulse solution $\bar{Z}_{p}$ at 
\begin{align}
\label{SN}
\begin{aligned}
0 & =  \frac{2 \sqrt{2}}{3} \left(e^{z^*} + e^{-z^*}  \right) - \frac{8 \alpha}{c^3 \phi_v^3} \left(  e^{z^*} + e^{-z^*} - 2 e^{-\phi_v z^*} - 2 \phi_v z^* e^{- \phi_v z^*} \right)   \\ & 
\quad
- \frac{8 \beta}{c^3 D^2 \phi_w^3} \left(  e^{z^*} + e^{-z^*} - 2 e^{-\phi_w z^*} - 2 \phi_w z^* e^{- \phi_w z^*} \right). 
\end{aligned}
\end{align}
\end{theorem}
\begin{theorem}
\label{TH:2P}
Let $(\hat{\tau}, \hat{\theta}) = (\varepsilon^2 \tau, \varepsilon^2 \theta) = (1, D^2)$ and let $(\alpha, \beta, \gamma, D)$ be such that 
\begin{align}
\label{existence_criterion_two}
\alpha \bar{V}_i(y_1,y_2,y_3) + \beta \bar{W}_i(y_1,y_2,y_3) + \ga + (-1)^i\frac{\sqrt{2}}{3} c  & =  0 \,,  \qquad i=1,\ldots,4\,,
\end{align}
has positive solutions $y_1, y_2, y_3$ and $c$. Then, 
for small $\varepsilon$ enough, (\ref{FHNd}) supports a traveling $2$-pulse solution $\bar{Z}_{2p} = (\bar{U}, \bar{V}, \bar{W})$ which goes 
asymptotically to $(U_b,U_b,U_b)$ as $x \to \pm \infty$, 
that travels with propagating speed $\varepsilon^2 c$, has leading order widths $2y_1/c$ and $2y_3/c$ and the distance between the two pulses are to leading order $2y_2/c$.
Here, $\bar{V}_i$ for $i = 1, \ldots, 4$, are the scaled values of the $V$-component at the four interfaces $z_i$ (so $2 y_i=z_{i+1}-z_i$), and these are, to leading order, given by
\begin{align} 
\label{VV2}
\begin{aligned}
\bar{V}_1(y_1, y_2, y_3) & =  \frac{1}{\phi_v} - \frac{1+\phi_v}{\phi_v} e^{(1-\phi_v) y_1} \left( 1 - e^{(1-\phi_v) y_2} ( 1 - e^{(1-\phi_v) y_3} ) \right)  \, , \\
\bar{V}_2(y_1, y_2, y_3) & =  \frac{1+\phi_v}{\phi_v} e^{(1-\phi_v) y_2} \left( 1 - e^{(1-\phi_v) y_3} \right)  + \frac{1-\phi_v}{\phi_v} e^{-(1+\phi_v) y_1} - \frac{1}{\phi_v}  \, , \\
\bar{V}_3(y_1, y_2, y_3) & = \frac{1}{\phi_v} - \frac{1+\phi_v}{\phi_v} e^{(1-\phi_v) y_3} - \frac{1-\phi_v}{\phi_v} e^{-(1+\phi_v) y_2}  \left( 1 - e^{-(1+\phi_v) y_1} \right)  \, ,  \\
\bar{V}_4(y_1, y_2, y_3) & = \frac{1- \phi_v}{\phi_v} e^{-(1+\phi_v) y_3} \left( 1 - e^{-(1+\phi_v) y_2}(1 - e^{-(1+\phi_v) y_1} ) \right)  - \frac{1}{\phi_v}   \,, 
\end{aligned}
\end{align}
where, again, $\displaystyle{\phi_v = \sqrt{1 + \frac{4}{c^2}}}$ and  
the values of the $W$-component at the interfaces are as (\ref{VV2}) but with $\phi_v$ replaced by $\displaystyle{\phi_w = \sqrt{1 + \frac{4}{c^2D^2}}}$.

In addition, a saddle node bifurcation occurs on the solution branch of traveling $2$-pulse solutions $\bar{Z}_{2p}$ at 
\begin{align}
\label{SN2}
\begin{aligned}
0&=- \frac{8 \alpha}{c^3 \phi_v^3} f(z_1, z_2, z_3, z_4; \phi_v) 
- \frac{8 \beta}{c^3 D^2 \phi_w^3} f(z_1, z_2, z_3, z_4; \phi_w)  \\
&\quad
+ \frac{2 \sqrt{2}}{3} (e^{z_1} + e^{z_2} + e^{z_3} + e^{z_4})\,, 
\end{aligned}
\end{align}
where 
\begin{align}
\label{func4}
\begin{aligned}
f(z_1, z_2, z_3, z_4; \phi) 
& =   
e^{z_1} \left( 1 - (1 + \phi y_1) e^{(1-\phi) y_1} + (1+\phi(y_1 + y_2)) e^{(1-\phi) (y_1+y_2)} 
\right.  \\ 
& \quad  \left.  - (1 + \phi(y_1 + y_2 + y_3)) e^{(1-\phi) (y_1 + y_2 + y_3) }  \right) 
  \\ 
& \quad
+ 
e^{z_2} \left( 1 - (1+\phi y_2) e^{(1-\phi) y_2}   + (1 + \phi (y_2 + y_3) ) e^{(1-\phi) (y_2+y_3)}  \right.
 \\ 
& \quad \left.
- (1 + \phi y_1 ) e^{-(1+\phi) y_1 }  \right)   + e^{z_3} \left( 1 - (1 + \phi y_3 ) e^{(1-\phi) y_3}  
\right. \\ & \quad \left.
+ (1 + \phi(y_1 + y_2)) e^{-(1+\phi) (y_1+y_2)} - (1+\phi y_2) e^{-(1+\phi) y_2}  \right) 
\\ & \quad
+ e^{z_4} \left( 1 - (1 + \phi(y_1 + y_2 + y_3)) e^{-(1+\phi) (y_1+y_2+y_3)}
\right.  \\  
& \quad \left. + (1 + \phi (y_2 + y_3)) e^{-(1+\phi) (y_2+y_3)} - (1 + \phi y _3) e^{-(1+\phi) y_3 }  \right).   
\end{aligned} 
\end{align}
\end{theorem}

We remark that the existence condition (\ref{existence_criterion_two}) for a traveling $2$-pulse solution encompasses the existence condition (\ref{existence_criterion_one}) for a traveling $1$-pulse solution. 
That is, upon taking the limit of $y_2 \to+ \infty$ in (\ref{existence_criterion_two}), i.e., upon letting the distance between the two pulses of the traveling $2$-pulse solution go to infinity, (\ref{existence_criterion_two}) separates into twice
(\ref{existence_criterion_one}). Once for $y_1$ and once for $y_3$ and both values approach $z^*$. 
This means that the coexistence of both traveling $1$-pulse and $2$-pulse solutions is guaranteed as discussed in \S\ref{SS:EXIST2}, see Fig.~\ref{fig06}. Numerical counterparts are shown in Fig.~\ref{fig07} and Fig.~\ref{fig08} in the final section. 
Another direct consequence of the existence condition (\ref{existence_criterion_two}) is the following.
\begin{lemma} \label{L:2P}
A necessary condition for the existence of a traveling $2$-pulse solution in (\ref{FHNd}) with $(\hat{\tau}, \hat{\theta}) = (\varepsilon^2 \tau, \varepsilon^2 \theta) = (1, D^2)$ is $\al\be<0$.
\end{lemma}
Note that this is also a necessary condition for the existence of a stationary $2$-pulse solution \cite{vH_EXIS}.
 
This article is organized as follows, in \S\ref{SS:PROF} we derive the profile of a traveling $1$-pulse solution 
with undetermined width and propagating speed. In \S\ref{SS:ACTION}, we compute the associate action functional 
for this profile and we derive the conditions for existence and saddle node bifurcation as stated in Theorem~\ref{TH:1P}. Observe that the existence result (\ref{existence_criterion_one}) has previously been derived in \cite{vH_EXIS}, while the 
results for the 
saddle node bifurcation is new. In \S\ref{SS:HYBD}, we derive the same conditions by using a different analytical approach utilizing the singular limit \cite[e.g.]{JJIAM,PhysD}. In \S\ref{SS:APPL}, we apply the methodology of the action functional to the problem for traveling $2$-pulse solution and derive the new results as stated in Theorem~\ref{TH:2P}. Furthermore, we deduce the necessary condition of Lemma~\ref{L:2P} by studying the existence condition (\ref{existence_criterion_two}).
We end the article with a summary and a discussion related to the collision of traveling pulse solutions and to Hopf bifurcations near the saddle node bifurcation, see \S\ref{SS:COL}. See also Remark~\ref{R:0}. 

\begin{remark}\label{R:0}
In \cite{vH_ACTION1}, we used the same methodology to study the existence and the stability of stationary pulse solutions for $\tau$ and $\theta$ of $\mathcal{O}(1)$. Here, we extend this methodology to the current setting of $\tau$ and $\theta$ large ($\mathcal{O}(1/\eps^2)$). However, in the current setting we cannot infer any stability results from the minimizers of the action functional since additional small eigenvalues pop out of the essential spectrum upon increasing $\tau$ and/or $\theta$ from $\mathcal{O}(1)$ to $\mathcal{O}(1/\eps^2)$ \cite{Martina, Martina2, vH_STAB} and these small eigenvalues are not tracked by the action functional approach. These additional small eigenvalues can of course destabilize the traveling pulse solutions, see for instance Fig.~\ref{fig08} in \S\ref{SS:COL}.
\end{remark}

\begin{remark} \label{R:1}
In \cite{CH5}, Chen et al. studied a geometric variational functional for a 
two-component FitzHugh--Nagumo model in which the activator $U$ is  
weakly coupled with inhibitor in the $U$-equations similar to the present paper with bistable case. In this article, we clearly highlight the importance and added complexity of having a third $W$-component in (\ref{FHNd}). However, our analysis for the three-component system can formally be reduced to cover the existence conditions for traveling pulse solutions for the corresponding two-component FitzHugh--Nagumo model 
\begin{align}
 \label{FHN2}
 \left\{
\begin{aligned}
    U_t & = \eps^2 U_{xx} + U - U^3 - \eps(\al V + \ga) \,,\\  
        \tau V_t & = V_{xx} + U - V. 
\end{aligned}
\right.
\end{align}
\end{remark}

\section{Traveling $1$-pulse solutions} \label{S2}
In this section, we prove the main result of this article related to the existence and saddle node bifurcation of traveling $1$-pulse solutions as stated in Theorem~\ref{TH:1P} and as shown in Fig.~\ref{fig01_NEW}.   
\subsection{The profile of a traveling $1$-pulse solution}
\label{SS:PROF}
Without loss of generality, we only consider right-going traveling $1$-pulse solutions, i.e., $c>0$, and
 we follow \cite{vH_EXIS, vH_ACTION1} to first determine the leading order profile of a right-going traveling $1$-pulse solution (with unknown width and speed). Note that we only show the crucial steps of this derivation and we refer to \cite{vH_EXIS, vH_ACTION1} for more details regarding the methodology. 

\begin{figure}
 \centering
 \includegraphics[width=7cm]{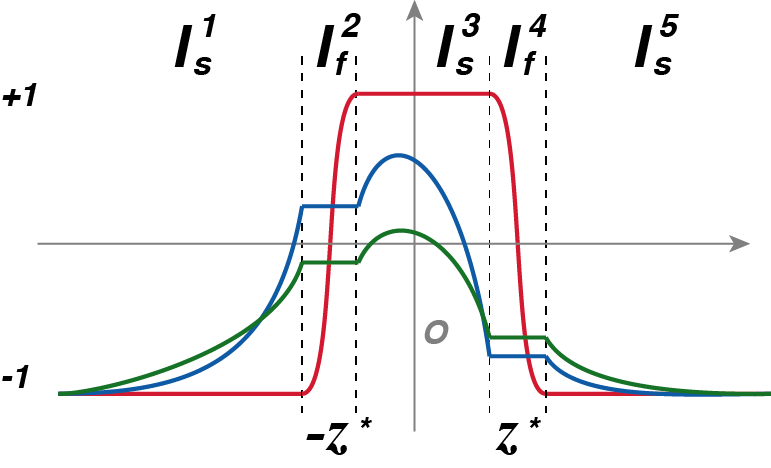}
  \caption{
    Schematic picture of the three slow region $I_s^{1,3,5}$ and two fast regions $I_f^{2,4}$
    introduced in (\ref{SLOWFAST})
 and used in this article to study traveling $1$-pulse solutions.
}
  \label{fig01_DOM}
\end{figure}

We divide the spatial domain into three slow regions $I_s^{1,3,5}$ and two fast regions $I_f^{2,4}$ to study (\ref{comoving}):
\begin{eqnarray}
\label{SLOWFAST}
\begin{array}{rl}
&I_s^1:=(-\infty,-z^*-\se]\,, \,\, I_f^2:=(-z^*-\se,-z^*+\se)\,,  \\
& I_s^3:=[-z^*+\se, z^*-\se]\,, \,\,
I_f^4:=(z^*-\se,z^*+\se)\,, \\
& I_s^{5}:=[z^*+\se, \infty)\,,  
\end{array}
\end{eqnarray}  
where $\pm z^*$ are the locations of the interfaces of the traveling pulse solution, 
that is, $\bar{U}(\pm z^*)=0$, see Fig.~\ref{fig01_DOM} and \cite{vH_EXIS, vH_ACTION1, vH_STAB}. 
We say that the width of the traveling pulse is given by $h := 2 z^*$. 
Rescaling by $\xi = (z \pm z^*)/ \eps$ depending on which fast region we are focusing, 
the first equation in (\ref{comoving}) for the $U$-component becomes 
\begin{eqnarray}
- \eps c^2 \bar{U}_\xi  & = c^2 \bar{U}_{\xi \xi} + \bar{U} - \bar{U}^3 - \eps (\alpha \bar{V} + \be \bar{W} + \ga)  \,. 
\label{comovingU}
\end{eqnarray} 
Upon using a regular expansion in $\eps$, $\bar{Z}_{p}(\xi) = \bar{Z}_0(\xi) + \O(\eps)$, we can, to leading order, analytically solve the above equation \cite[e.g.]{vH_EXIS}. From this we determine the leading order profile of $\bar{U}$: 
\begin{eqnarray}
\label{Uzero}
&\bar{U}_{0} (z)= \left\{
\begin{array}{rl}
 -1\,, \qquad& z\in I_s^1\,,\\
\displaystyle{ \tanh{\left(\frac{z+z^*}{\sqrt2 \eps c}\right)} } \,, \qquad& z\in I_f^2\,,\\
 1\,, \qquad& z\in  I_s^3\,,\\
\displaystyle{ -\tanh{\left(\frac{z-z^*}{\sqrt2 \eps c}\right)} } \,, \qquad& z\in I_f^4\,,\\
-1\,, \qquad& z\in I_s^{5}  \,.
\end{array}
\right.\, 
\end{eqnarray}
To leading order the $V$- and $W$-components are constant over the two fast fields \cite[e.g.]{vH_EXIS}, and 
we rewrite the second and third equations in (\ref{comoving}) for the $V$- and $W$-components to determine their profiles in the slow fields 
\begin{eqnarray}
- c^2 \bar{V}_z   =  c^2 \bar{V}_{zz} + \bar{U} - \bar{V} \,, \qquad
- D^2 c^2 \bar{W}_z   =  D^2 c^2 \bar{W}_{zz} + \bar{U} - \bar{W} \,,
 \label{comovingV} 
\end{eqnarray} 
which satisfy $\bar{V} = \mathcal{L}_{1c} \bar{U}$ and $\bar{W} = \mathcal{L}_{2c} \bar{U}$ with the self-adjoint operators $\mathcal{L}_{1c}$ and  
$\mathcal{L}_{2c}$ (with $\hat{\theta}=D^2$) defined in (\ref{OPER}).  

Upon using a regular expansion in $\eps$, 
$\bar{Z}_{p}(\xi) = \bar{Z}_0(\xi) + \O(\eps)$ with $\bar{Z}_0 = (\bar{U}_0, \bar{V}_0, \bar{W}_0)$, we solve the linear equations (\ref{comovingV}) %
\begin{eqnarray*}
&\bar{V}_{0}(z) = \left\{
\begin{array}{rl}
\displaystyle{ \frac{2(1+\phi_v)}{\phi_v} e^{-\frac{1-\phi_v}{2} z} \sinh \left(-\frac{1-\phi_v}{2} z^* \right) - 1 }
\,, \qquad&z\in I_s^1\,,\\
\displaystyle{ \frac{1}{\phi_v} \left( 1 - (1+\phi_v) e^{(1-\phi_v) z^*} \right) } \,, \qquad& z\in  I_f^2\,,\\
\displaystyle{ - \frac{1+\phi_v}{\phi_v} e^{-\frac{1-\phi_v}{2} (z-z^*)} 
+ \frac{1-\phi_v}{\phi_v} e^{-\frac{1+\phi_v}{2} (z+z^*)} + 1}\,, \qquad&z\in I_s^3\,,\\
 \displaystyle{ \frac{1}{\phi_v} \left( (1-\phi_v) e^{-(1+\phi_v) z^*} - 1 \right) } \,, \qquad& z\in  I_f^4  \,, \\
\displaystyle{ \frac{2(1-\phi_v)}{\phi_v} e^{-\frac{1+\phi_v}{2} z} \sinh \left( -\frac{1+\phi_v}{2} z^* \right) - 1 } \,, \qquad&z\in I_s^{5} \,, 
\end{array}
\right.
\end{eqnarray*}
where $\displaystyle{\phi_v = \sqrt{1 + \frac{4}{c^2}}}$. 
The profile of $\bar{W}_0(z)$ is obtained by replacing $\phi_v$ by $\phi_w = \displaystyle{\sqrt{1 + \frac{4}{c^2D^2}}}$ in $\bar{V}_0(z)$.
A typical profile of a traveling $1$-pulse solution is given in Fig. \ref{fig01_NEW}(a).

\subsection{The action functional}
\label{SS:ACTION}
Next we use the action functional approach \cite[e.g.]{CH2, vH_ACTION1} to determine the width, speed, and stability of a traveling $1$-pulse solution.
The action functional for a traveling pulse solution is similar to 
the action functional for the standing pulse \cite{CH1, vH_ACTION1}. 
In particular, the action functional for a traveling $1$-pulse solution $\bar{Z}_{p} = (\bar{U},\bar{V}, \bar{W})$ 
-- whose profile with unknown width $h:=2 z^*$ and propagating speed $c$ 
have been computed in the previous section -- is given by 
\begin{align}
\label{Jc}
\begin{aligned}
J_c(u)
& =  \int_{-\infty}^{\infty} e^{z} \left( \frac{ \varepsilon^2 c^2}{2} 
u_z^2 + F(u) - F(U_b)+ \frac{ \varepsilon \al}{2} (u \mathcal{L}_{1c} u-U_b^2)
\right. \\ & \qquad + \frac{ \varepsilon \be}{2} (u \mathcal{L}_{2c} u-U_b^2)  + \varepsilon \ga (u-U_b) \Big) dz \,,  
\end{aligned}
\end{align}
with $F(u) := u^4/4 - u^2/2$,  and $U_b$ begin the zero of 
$ u^3-u+\eps( (\al + \be) u + \ga) $ near $u \simeq -1$. 
Thus $U_b = -1 + \mathcal{O}(\varepsilon)$ and $(U_b, U_b, U_b)$ is the constant 
steady states attained by the traveling 1-pulse solution at both ends.  
We introduce the Hilbert space $\mathbb{H}_{ex}^1$ corresponding to the inner product 
$\displaystyle{ \langle v,w \rangle_{\mathbb{H}_{ex}^1} = \int_\mathbb{R} e^x \left( v w + v_x w_x \right) dx}$. 
The variational approach will find the weak solutions in $\mathbb{H}_{ex}^1$ to (\ref{comoving}) and 
a class $\mathcal{A}$ of admissible functions is defined as $\displaystyle{\mathcal{A} \equiv \{ u - U_b \in \mathbb{H}_{ex}^1 \} }$ \cite{CH1}. 
In other words, we consider the functional $J_c : \mathcal{A} \to \mathbb{R}$ for $c > 0$. 

The Gateaux derivative of $J_c$ is  
\begin{align*}
\frac{\delta J_c}{\delta u} \Phi &=  \lim \limits_{t \to 0} \frac{J_c(u + t \Phi) - J_c(u)}{t} \,  \\
&= \int_{-\infty}^{\infty} e^{z} \left( \varepsilon^2 c^2 u_z \Phi_z + 
 (u^3 -  u) \Phi + \frac{\varepsilon \al}{2} (\Phi \mathcal{L}_{1c} u + u \mathcal{L}_{1c} \Phi) 
\right. \\&\qquad + \frac{\varepsilon \be}{2} (\Phi \mathcal{L}_{2c} u + u \mathcal{L}_{2c} \Phi)
 + \varepsilon \ga \Phi \Big) dz\,,  \\ 
&= \int_{-\infty}^{\infty} e^{z} \left(- \varepsilon^2 c^2 u_{zz} 
- \varepsilon^2 c^2 u_{z} + u^3 - u + \varepsilon (\al  \mathcal{L}_{1c} u + \be  \mathcal{L}_{2c} u + \ga)  \right) \Phi dz\,, 
\end{align*}
with the self-adjoint operators $\mathcal{L}_{1c}$ and $\mathcal{L}_{2c}$ (with $\hat{\theta}=D^2$) defined in (\ref{OPER}).
Thus, we find that $\displaystyle{\frac{\delta J_c}{\delta u} \Phi =0}$ for all $\Phi \in C_0^{\infty}$ if 
$\bar{U}$ is the weak solution of the equation  
\begin{align*}
\varepsilon^2 c^2 \bar{U}_{zz} + \varepsilon^2 c^2 \bar{U}_{z} - \bar{U}^3 + \bar{U} - 
\varepsilon (\al  \mathcal{L}_{1c} \bar{U} + \be  \mathcal{L}_{2c} \bar{U} + \ga ) = 0. 
\end{align*}
That is, the critical points of $J_c$ 
satisfy the Euler-Lagrange equation associated with $J_c$ and these 
coincide with the traveling $1$-pulse solutions for (\ref{comoving}) 
when we set $\bar{V} = \mathcal{L}_{1c} \bar{U}$ and  $\bar{W} = \mathcal{L}_{2c} \bar{U}$.

Let $\Xi[a]$ be the translation operator along the $z$-axis for a distance $a \in \mathbb{R}$ 
given by $\Xi[a](\bar{Z}_p(z)) = \bar{Z}_{p}(z-a)$. 
It follows that 
\begin{align*}
J_c(\Xi[a] (\bar{Z}_{p}))
& =    \int_{-\infty}^{\infty} e^{z} \left( \frac{ \varepsilon^2 c^2}{2} 
(\Xi[a] U_z)^2 + F(\Xi[a] U) - F(U_b) 
\right.  \\
&  \qquad 
+ \frac{ \varepsilon \al}{2} ((\Xi[a] U) \mathcal{L}_{1c} (\Xi[a] U)-U_b^2)  
+ \frac{ \varepsilon \be}{2} ((\Xi[a] U) \mathcal{L}_{2c} (\Xi[a] U)-U_b^2) 
  \\
&  \qquad
+ \varepsilon \ga ((\Xi[a] U) -U_b) \Big) dz \,,   \\
& =    \int_{-\infty}^{\infty} e^{z+a} \left( \frac{ \varepsilon^2 c^2}{2} 
U_z^2 + F(U) - F(U_b) +  \frac{ \varepsilon \al}{2} (U \mathcal{L}_{1c} U-U_b^2)  \right.   \\
& \qquad  
+  \frac{ \varepsilon \be}{2} (U(\xi) \mathcal{L}_{2c} U(\xi)-U_b^2) 
+ \varepsilon \ga (U(\xi) -U_b) \Big) dz \,\\&=  e^a J_c(\bar{Z}_p).    
\end{align*}
%
Any spatial translation of a traveling pulse remains a traveling pulse, this leads to a one-dimensional manifold of translated solutions. Suppose the critical point $\bar{Z}_p$ is smooth and both it and its derivative decay sufficiently fast as $z \to \pm \infty$, 
by setting the test function $\Phi = \bar{U}_z$, integration by parts leads to  
\begin{align*}
J_c(\bar{Z}_{p}) & =  \left[ e^{z} \left( \frac{ \varepsilon^2 c^2}{2} \bar{U}_z^2 + F(\bar{U}) - F(U_b) 
+ \frac{ \varepsilon \al}{2} (\bar{U} \mathcal{L}_{1c} \bar{U}-U_b^2) \right. \right.  \\
& \qquad
+ \frac{ \varepsilon \be}{2} (\bar{U} \mathcal{L}_{2c} \bar{U}-U_b^2) 
+ \varepsilon \ga (\bar{U}-U_b) \Big) \Big]_{-\infty}^\infty = 0.
\end{align*}
due to the assumed asymptotic behavior of $\bar{U}$ and its derivative as $z \to \pm \infty$. Hence besides $\bar{Z}_p$ being a critical point of $J_c$, we have in addition $J_c(\bar{Z}_{p})= 0$. 
\begin{lemma}
The action functional $J_c$ of a traveling pulse solution $\bar{Z}_p$, and its derivative with respect to $z^*$,  are given by
\begin{align}
\label{Jzero}
\begin{aligned}
 \frac{J_c (\bar{Z}_{p})}{\varepsilon}  &= 
 \frac{2 \alpha}{\phi_v}
\left(  2 e^{-\phi_v z^*} - e^{z^*} - e^{-z^*}\right) + 
 \frac{2 \beta}{\phi_w}
\left(  2 e^{-\phi_w z^*} - e^{z^*} - e^{-z^*}\right)   \\
& \quad
+ 2 \ga (e^{z^*} - e^{-z^*})
+\frac{2 \sqrt{2}}{3} c (e^{z^*} + e^{-z^*}) + \mathcal{O}(\se)
  \,,  
  \end{aligned}
\end{align}
and 
\begin{align}
\label{DJzero}
\begin{aligned} \frac{1}{\varepsilon} \frac{\partial }{\partial z^*} J_c(\bar{Z}_{p}) &= 
 \frac{2 \alpha }{\phi_v}
\left( e^{-z^*} - e^{z^*} - 2 \phi_v e^{-\phi_v z^*} \right) + 
\frac{2 \beta }{\phi_w}
\left( e^{-z^*} - e^{z^*} - 2 \phi_w e^{-\phi_w z^*} \right)   \\
& \quad
+ 2 \ga (e^{z^*} + e^{-z^*}) +\frac{2 \sqrt{2}}{3} c (e^{z^*} - e^{-z^*}) + \mathcal{O}(\se)
 \,.
\end{aligned}
\end{align}
\end{lemma}

\begin{proof}
To prove the lemma, 
we split the definite integral of $J_c$ into the five regions   
\begin{align*}
J_c (\bar{Z}_{p}) =   \int_{ I_s^1} + \int_{ I_f^2} + \int_{ I_s^3} + \int_{ I_f^4} + \int_{ I_s^5} \,.
\end{align*}
Upon using that $\bar{U}_0=-1$ in the slow regions $I_s^{1,5}$ (\ref{Uzero}), we get, to leading order,  
\begin{align*}
 \int_{ I_s^1} + \int_{ I_s^3} + \int_{ I_s^5}  
& =  -  \varepsilon \int_{-\infty}^{-z^* -\se} e^z \left( \alpha
\frac{1+\phi_v}{\phi_v} e^{-\frac{1-\phi_v}{2} z} \sinh \left( -\frac{1-\phi_v}{2} z^* \right) \right.   \\
& \quad \displaystyle{ \left. + \beta
\frac{1+\phi_w}{\phi_w} e^{-\frac{1-\phi_w}{2} z} \sinh \left( -\frac{1-\phi_w}{2} z^* \right)  
 \right) dz } \, \\
& \quad  +  \varepsilon \int_{-z^* +\se}^{z^* -\se} e^z \left( \frac{\alpha}{2}\left(
\frac{1-\phi_v}{\phi_v} e^{-\frac{1+\phi_v}{2} (z+z^*)}  \right. \right. \\
& \quad \left. - \frac{1+\phi_v}{\phi_v} e^{-\frac{1-\phi_v}{2} (z-z^*)}\right)    
+ \frac{\beta}{2}\left(
\frac{1-\phi_w}{\phi_w} e^{-\frac{1+\phi_w}{2} (z+z^*)} \right. \\
& \quad
\left. \left. - \frac{1+\phi_w}{\phi_w} e^{-\frac{1-\phi_w}{2} (z-z^*)} \right) 
+ 2 \ga \right) dz \\
& \quad  - \varepsilon \int_{z^* +\se}^{\infty} e^z \left( \alpha
\frac{1-\phi_v}{\phi_v} e^{-\frac{1+\phi_v}{2} z} \sinh \left( -\frac{1+\phi_v}{2} z^* \right) \right.  \\
& \quad \left. + \beta
\frac{1-\phi_w}{\phi_w} e^{-\frac{1+\phi_w}{2} z} \sinh \left( -\frac{1+\phi_w}{2} z^* \right) 
 \right) dz  \\ 
& =    \frac{2  \varepsilon \al}{\phi_v} \left( 2 e^{-\phi_v z^*} - e^{z^*} - e^{-z^*}  \right)
\\ &\quad+  \frac{2  \varepsilon \be}{\phi_w} \left( 2 e^{-\phi_w z^*} - e^{z^*} - e^{-z^*}  \right)+
2 \varepsilon \ga (e^{z^*} - e^{-z^*}) \,.  
\end{align*}
As for two fast regions $I_f^2$ and $I_f^4$ with $\xi = (z + z^*)/ \varepsilon$ 
and $\xi = (z - z^*)/\varepsilon$, respectively, we get, to leading order, 
\begin{align*}
 \int_{ I_f^2} + \int_{ I_f^4}  
& =  \varepsilon \int_{-\xi_* - 1/ \se}^{-\xi_* + 1/ \se} e^{\varepsilon \xi}
\left( \frac{1}{4} {\rm sech}^4 \left( \frac{\xi + \xi_*}{\sqrt{2} c} \right) + \frac{1}{4} \tanh^4 \left( \frac{\xi + \xi_*}{\sqrt{2} c} \right) \right.\\
&\quad \left.
- \frac{1}{2} \tanh^2 \left( \frac{\xi + \xi_*}{\sqrt{2} c} \right)  + \frac{1}{4} \right) d \xi  
\\ &\quad + \ \varepsilon \int_{\xi_* - 1/ \se}^{\xi_* + 1/ \se} e^{\varepsilon \xi}
\left( \frac{1}{4} {\rm sech}^4 \left( \frac{\xi - \xi_*}{\sqrt{2} c} \right) 
+ \frac{1}{4} \tanh^4 \left( \frac{\xi - \xi_*}{\sqrt{2} c} \right) \right.
\\ &\quad  \left. - \frac{1}{2} \tanh^2 \left( \frac{\xi - \xi_*}{\sqrt{2} c} \right)  + \frac{1}{4} \right) d \xi  \\
& = \frac{2 \sqrt{2}}{3} \varepsilon c (e^{\varepsilon \xi_*} + e^{-\varepsilon \xi_*}) .
\end{align*}
Combining these integrals gives (\ref{Jzero}), and subsequently taking the derivative with respect to $z^*$ gives
(\ref{DJzero}).
\end{proof}

As $\bar{Z}_p$ is a critical point with $J_c(\bar{Z}_{p}) = 0$, we can set the left-hand sides of (\ref{Jzero}) and (\ref{DJzero}) to zero to obtain, 
\begin{align}
\label{COND1}
\begin{aligned}
0 = \frac{1}{\varepsilon}\frac{\partial J_c}{\partial z^*} &= 
 4 e^{z^*} \left( \frac{\alpha}{\phi_v}  (e^{-(1+\phi_v) z^*} - 1) + 
\frac{\beta}{\phi_w}  (e^{-(1+\phi_w) z^*} - 1) + \ga + \frac{\sqrt{2}}{3} c \right)
\\
& =  4 e^{z^*} \left( \alpha \bar{V}_0(z^*) + 
\beta \bar{W}_0(z^*) + \ga + \frac{\sqrt{2}}{3} c \right).
\end{aligned}
\end{align}
The system inherits the symmetry $(z^*,c) \to (-z^*, -c)$. 
That is, whenever there is a right-going traveling pulse solution ($c>0$) there is also a left-going traveling pulse solution ($c < 0$), since the traveling pulse solutions do not have a preferred direction. Recalling $J_c(\bar{Z}_{p}) = 0$ (\ref{Jzero}) and $\displaystyle{\frac{\partial}{\partial z^*} J_c(\bar{Z}_{p}) =0}$ (\ref{DJzero}) again, we obtain a second condition similar to (\ref{COND1}). 
\begin{align*}
4 e^{-z^*} \left( \alpha \bar{V}_0(-z^*) + \beta \bar{W}_0(-z^*) + \ga - \frac{\sqrt{2}}{3} c \right) = 0. 
\end{align*}
Combining these two conditions yields the existence conditions for a traveling $1$-pulse solution in terms of the two undetermined variables $z^*$ and $c$, 
\begin{align}
\label{existence_criterion}
\begin{aligned} 
0&= \alpha \bar{V}_0(z^*) + \beta \bar{W}_0(z^*) + \ga + \frac{\sqrt{2}}{3} c \,,  \\
0&= \alpha \bar{V}_0(-z^*) + \beta \bar{W}_0(-z^*) + \ga - \frac{\sqrt{2}}{3} c   \,, 
\end{aligned}
\end{align}
which is the same as (\ref{existence_criterion_one}) of Theorem~\ref{TH:1P}.

By solving (\ref{existence_criterion_one})/(\ref{existence_criterion}) for $z^*$ and $c$, we get the solution branches with respect to $D^2$ as shown in Fig. \ref{fig02} for $(\alpha, \beta, \gamma) = (2,1,1)$. 
At $D^2 \approx 1.315$, traveling pulse solutions are emanated from 
the standing pulse solutions in a subcritical manner. 
The asymptotic behavior of the solution branches for large $D^2$ approaches 
$c_\infty$ which satisfies $\displaystyle{ \frac{\sqrt{2} c}{3}  =  \frac{\alpha}{\phi_v} + \beta - \gamma}$. For instance, $c_\infty = \sqrt{14}$ for $(\alpha, \beta, \gamma) = (2,1,1)$.
\begin{figure}[ht!]
 \centering
 \includegraphics[width=10cm]{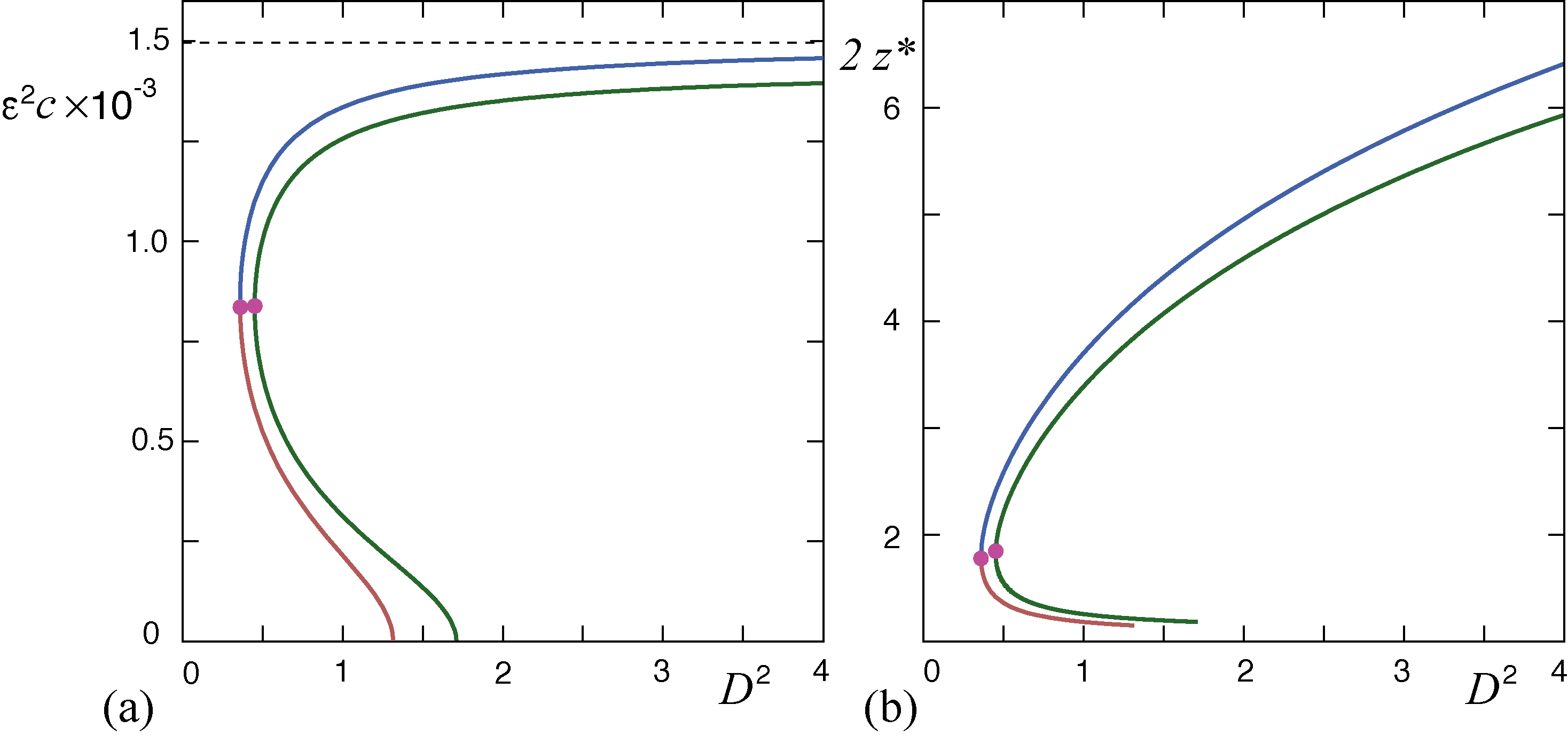}
  \caption{ The bifurcation diagram with respect to $D^2$: 
     the vertical axes in (a) and (b) are the propagating velocity $c$ 
     and the pulse width $2 z^*$, respectively. The system parameters are set to $(\al, \be, \ga) = (2, 1, 1)$. Blue and red curves indicate solution branches obtained by solving the existence conditions (\ref{existence_criterion}) with $\displaystyle{ \frac{\partial}{\partial c} J_c(\bar{Z}_{p}) >0}$ and  $<0$, respectively. The resulting values $c$ and $2 z^*$  were rescaled to the original scale of $\varepsilon^2 c$ and $2 z^*/c$.
     The green curves are obtained by numerical continuation of the original PDE with $\eps = 0.02$, see Appendix~\ref{SS:NUM} for more details on the numerical continuation.  
The pink disks indicate the locations of turning points of solution branches. 
     In the left panel, the dotted line indicates $c_\infty= \varepsilon^2 \sqrt{14}$. 
}
  \label{fig02}
\end{figure}

As for the stability of the traveling pulse solutions (however, see Remark~\ref{R:0}), we first account for the second derivative with respect to $z^*$ 
\begin{align*}
 \frac{1}{\varepsilon} \frac{\partial^2 }{\partial (z^*)^2} J_c(\bar{Z}_{p}) &= 
\frac{16}{c^2} \left( \frac{\alpha}{\phi_v} e^{-\phi_v z^*} + \frac{\beta}{\phi_w D^2} e^{-\phi_w z^*}  \right), 
\end{align*} 
where we make use of $J_c(\bar{Z}_{p}) = 0$ in (\ref{Jzero}). When both $\alpha$ and $\beta$ are set to be positive, this expression 
is always positive (which is related to stable eigenvalues \cite{vH_ACTION1}) even for the solutions on the lower branch of saddle-node structure  in Fig. \ref{fig02}. 
Next, we consider the derivative with respect to $c$
\begin{align}
\label{CJzero}
\begin{aligned}
\frac{1}{\varepsilon} \frac{\partial}{\partial c}  J_c(\bar{Z}_{p}) &=
- \frac{8 \alpha}{c^3  \phi_v^3}
\left( e^{z^*} + e^{-z^*} - 2 e^{-\phi_v z^*} - 2 \phi_v z^* e^{-\phi_v z^*} \right)  \\
& \quad - \frac{8 \beta}{c^3 D^2 \phi_w^3}
\left( e^{z^*} + e^{-z^*} - 2 e^{-\phi_w z^*} - 2 \phi_w z^* e^{-\phi_w z^*} \right)  \\
& 
\quad +\frac{2 \sqrt{2}}{3} (e^{z^*} + e^{-z^*})  \,. 
\end{aligned}
\end{align}
By solving $\displaystyle{ \frac{\partial}{\partial c} J_c(\bar{Z}_{p})= 0}$ combined with $J_c(\bar{Z}_{p}) = 0$ 
and $\displaystyle{\frac{\partial}{\partial z^*} J_c(\bar{Z}_{p})=0}$,  we can detect the 
saddle-node bifurcation point as $D^2 \approx 0.359922$ with $(z^*, c) \approx (1.90119, 2.12014)$ 
for  $(\al, \be, \ga) = (2, 1, 1)$. The derivative $\displaystyle{ \frac{\partial}{\partial c} J_c(\bar{Z}_{p})}$ changes its sign 
from minus to plus at the turning point of the solution branch. 
Therefore, the traveling $1$-pulse solution, as long as the remaining small eigenvalues coming from the essential spectrum still have negative real part, see \cite{vH_STAB}, \S\ref{SS:COL} and Remark~\ref{R:0}, recovers their stability via a saddle-node bifurcation. 
Figure 4 show  the contour plots of the leading order component
 $J_c(\bar{Z}_{p})/ \varepsilon$ 
for $(\al, \be, \ga) = (2, 1, 1)$. For $D^2 = 0.80$ as in Fig. \ref{fig03}(a), 
the existence condition is solved by $(z^*, c) \approx (3.71331, 2.87201)$ 
and $(0.892803, 1.31111)$. The first and second solutions satisfy 
$\displaystyle{ \frac{\partial}{\partial c} J_c(\bar{Z}_{p}) >0}$ (node) and  $<0$ (saddle), respectively.  
In the neighborhood of a turning point, for example for $D^2 = 0.361$ as in Fig. \ref{fig03}(b),
the two curves of  $\displaystyle{\frac{\partial}{\partial z^*} J_c(\bar{Z}_{p})=0}$ and 
$\displaystyle{ \frac{\partial}{\partial c} J_c(\bar{Z}_{p}) = 0}$ intersect around 
the bottom of the basin at $\displaystyle{J_c(\bar{Z}_{p})=0}$. 
\begin{figure}
  \begin{center}
  \scalebox{0.60}{
 \includegraphics{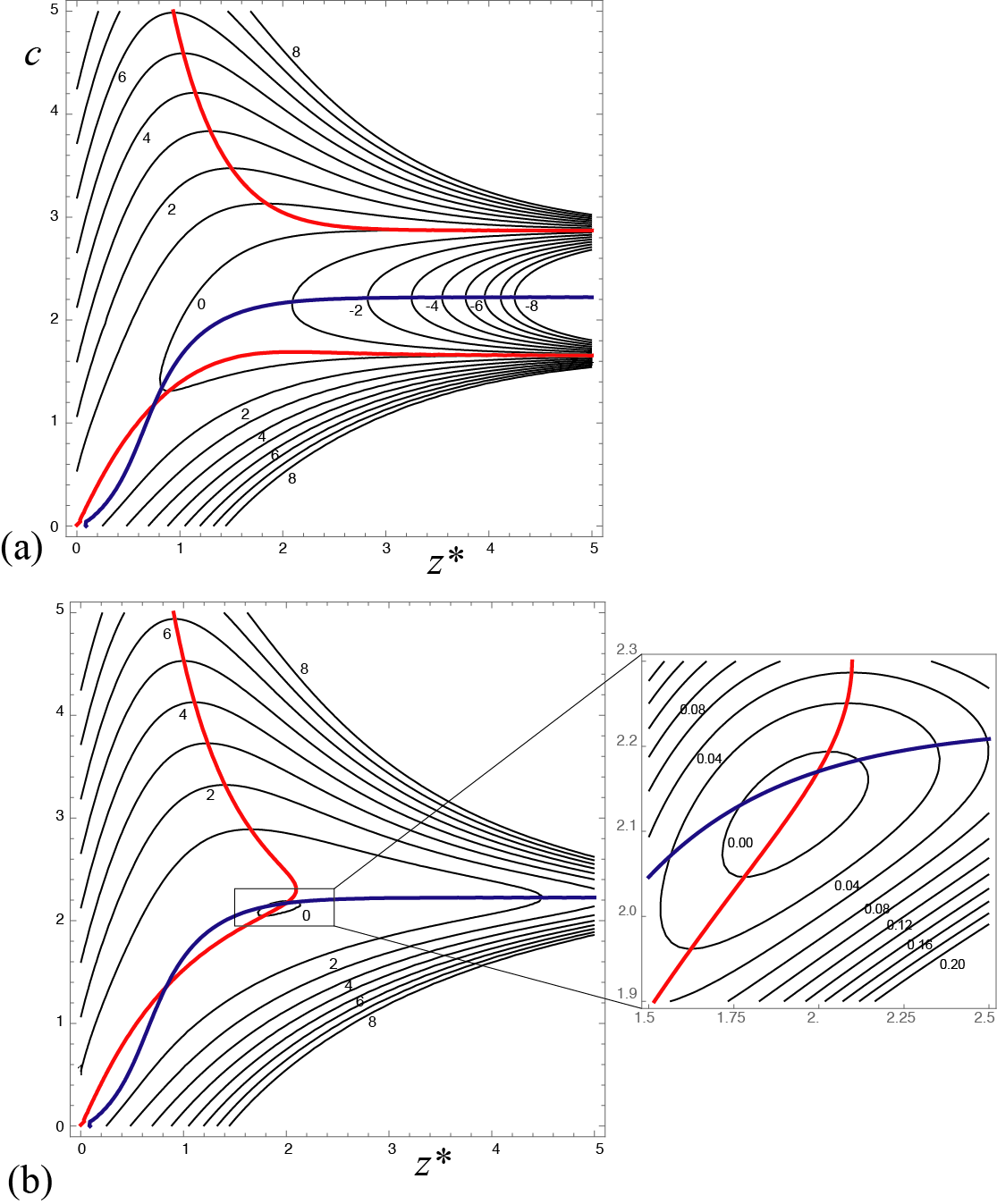}}
    \vspace*{-0.5cm}
    \caption{(a)(b) Contour plots of the leading order component of $J_c(\bar{Z}_{p})/ \varepsilon$ 
    for $D^2 = 0.80$ and $0.361$, respectively (black lines).  The system parameters are set to $(\al, \be, \ga) = (2, 1, 1)$.
    The red curve is determined by $ \partial J_c(\bar{Z}_{p})/ \partial z^*= 0$, while the blue line is $ \partial J_c(\bar{Z}_{p})/ \partial c= 0$. 
    The lower right panel shows a zoom in around the intersection of $J_c(\bar{Z}_{p}) = 0$, 
    $ \partial J_c(\bar{Z}_{p})/ \partial z^*= 0$ and $ \partial J_c(\bar{Z}_{p})/ \partial c= 0$ near $(z^*, c) \approx (1.90119, 2.12014 )$. 
     }
\label{fig03}
  \end{center}
\end{figure} 

We conclude this subsection by looking the following necessary parameter condition 
on the solvability of $z^*$ in $\displaystyle{ \frac{\partial}{\partial c} J_c(\bar{Z}_{p})= 0}$
with (\ref{CJzero}). 
\begin{lemma}
Let $c > 0$ and $z^* \in (0, \infty]$ satisfying (\ref{existence_criterion_one}) and (\ref{SN}). Then, 
\begin{align*}
\displaystyle{\frac{\sqrt{2} c}{12}} & 
\displaystyle{ <  \frac{\alpha}{(c^2 + 4) \phi_v} + \frac{\beta}{(c^2 D^2 + 4) \phi_w}. }
\end{align*} 
\end{lemma}
We define the function $\mathcal{F}(z^*,c)$, from (\ref{CJzero}) of 
the derivative with respect to $c$, as follows: 
\begin{eqnarray}
\displaystyle{ \frac{1}{\varepsilon} \frac{\partial}{\partial c} } J_c(\bar{Z}_{p}) &=& 
=: \displaystyle{ e^{z^*} \mathcal{F}(z^*, c)} \,. \nonumber
\end{eqnarray}
Note that that $\displaystyle{ \mathcal{F}(0,c) = \frac{4 \sqrt{2}}{3} } > 0$ and 
$\displaystyle{  \mathcal{F}(\infty,c) = \frac{2 \sqrt{2}}{3} - \frac{8 \alpha}{c^3 \phi_v^3}
- \frac{8 \beta}{c^3 D^2 \phi_w^3}   }$. 

Next we consider the derivative of $\mathcal{F}$ with respect to $z^*$, 
\begin{eqnarray}
\displaystyle{ \frac{\partial}{\partial z^*} \mathcal{F}(z^*,c) } &=& 
\displaystyle{ - 2 e^{-2 z^*}
\left( \frac{2 \sqrt{2}}{3}
-  \frac{8 \alpha}{c^3\phi_v^3} (1 - e^{(1-\phi_v) z^*} (1 + (1 + \phi_v) \phi_v z^*) ) \right.
} \nonumber \\
& & \displaystyle{ \left. 
-  \frac{8 \beta}{c^3 D^2 \phi_w^3} (1 - e^{(1-\phi_w) z^*} (1 + (1 + \phi_w) \phi_w z^*) )
\right). } \nonumber
\end{eqnarray}
Here $\displaystyle{ \frac{\partial}{\partial z^*} \mathcal{F}(0,c) = - \frac{4 \sqrt{2}}{3}  < 0} $, 
and $\displaystyle{ \frac{\partial}{\partial z^*} \mathcal{F}(\infty,c) =  
\lim \limits_{z^* \to \infty} \frac{\partial}{\partial z^*} \mathcal{F}(z^*,c) = 0}$. 
It is easy to see that the sign of $\displaystyle{ \frac{\partial}{\partial z^*} \mathcal{F}(\infty,c) }$ depends on that of $\mathcal{F}(\infty,c)$.

If $\displaystyle{ \mathcal{F}(\infty,c) > 0} $, $\mathcal{F}$ is monotonically decreasing and both $\mathcal{F}(0,c)$ and $\mathcal{F}(\infty,c)$ are positive. Then $\mathcal{F} > 0$ for all $z^*$. 
If $\displaystyle{ \mathcal{F}(\infty,c) 
< 0 }$, by the intermediate value theorem, there exist one positive 
$z^*$ such that $\mathcal{F}(z^*,c) = 0$. We can show that $\mathcal{F}$ 
reach a negative minimum at the non-negative root of $\displaystyle{\frac{\partial \mathcal{F}}{\partial z^*}  = 0}$. Then it increases again and converges to $\mathcal{F}(\infty,c) < 0$. This concludes the proof of the lemma. 

\subsection{Comparison with singular limit analysis}
\label{SS:HYBD}
Here, we derive the same conditions for the existence and saddle node bifurcation of a traveling $1$-pulse solution as in Theorem~\ref{TH:1P} by analyzing the singular limit of (\ref{FHNd}) in more detail, see \cite[e.g.]{JJIAM,PhysD} for more details on this technique. 
In other words, we provide another (sketch of a) proof of Theorem~\ref{TH:1P} showcasing the similarities and complementary character of the two techniques.

The $U$-component of a traveling pulse solution satisfies (\ref{comovingU}) 
in the comoving frame. 
Introducing the regular expansions $\tilde{U} = \tilde{U}_0 + \eps \tilde{U}_1 + \eps^2 \tilde{U}_2 \ldots$ 
and $c = c_0 + \eps c_1 + \eps^2 c_2 + \ldots$, and equating equal terms with respect to $\eps$, 
the leading order equation becomes $0 = c_0^2 \tilde{U}_{0 \xi \xi} + \tilde{U}_0 - \tilde{U}_0^3$. 
Solving it, we obtain $\displaystyle{ \tilde{U}_0 = \pm \tanh \left( \xi/ (\sqrt{2} c_0) \right) }$. 
For the next order ${\mathcal O}( \eps)$, we have 
\begin{align*}
-c_0^2 \tilde{U}_{0 \xi} & = c_0^2 \tilde{U}_{1 \xi \xi} + \tilde{U}_1 - 3 \tilde{U}_0^2 \tilde{U}_1 - (\alpha V + \beta W + \gamma),   \\
& =:  {\mathcal A}(\tilde{U}_1) - (\alpha V + \beta W + \gamma). 
\end{align*}
Note that ${\mathcal A}$ is self-adjoint, the derivative $\tilde{U}_{0 \xi}$ satisfies ${\mathcal A}(\tilde{U}_{0 \xi}) = 0$, and $(\alpha V + \beta W + \gamma)$ is evaluated at either $z^*$ or $-z^*$. 
Taking the inner product with $\tilde{U}_{0 \xi}$ and applying Fredholm's alternative 
$$\langle \tilde{U}_{0 \xi}, {\mathcal A}(\tilde{U}_1)) \rangle = 0$$ to obtain the solvability condition \cite{ToddKeith},  we find the traveling front solution up to the leading order as 
\begin{eqnarray}
\displaystyle{ 
\tilde{U} = \pm \tanh \left( \frac{x - \varepsilon^2 c t \pm z^*/c}{\sqrt{2} \eps}  \right), \qquad c = \pm \frac{3}{\sqrt{2}} \left. (\alpha V + \beta W + \gamma) \right |_{\mp z^*}. 
}
\end{eqnarray}
It is remarked that the propagating velocity must be in the order of ${\mathcal O}(\eps^2)$. 

We consider a traveling pulse solution as a solution which consists of a front and back and the positions of 
the interfaces are given by $l_2$ and $l_1$, respectively (with $l_1<l_2$). In the singular limit of $\eps \rightarrow 0$,  
the rectangular shape of $U$-component of the traveling pulse profile is replaced by a piecewise constant 
function $U(x; l_2, l_1) = F(x-l_2) - F(x-l_1) - 1$ with $F(x) = 1$ for $x \le 0$ and $-1$ for $x > 0$. 
Thus we obtain the following mixed ODE-PDE system, associated with (\ref{FHNd}), describing the dynamics of a traveling pulse solution
\begin{align}
\left \{
\begin{aligned}
\dot{l}_1 &= \frac{3 \eps^2}{ \sqrt{2}} (\alpha V (l_1) + \beta W (l_1)+\gamma), \\
 \dot{l}_2 &= -\frac{3 \eps^2}{ \sqrt{2}} (\alpha V(l_2) + \beta W(l_2) + \gamma),   \\
\tau V_t &= V_{xx} + U - V,  \\
\theta W_t &= D^2 W_{xx} + U - W. 
\end{aligned}
\right.
\label{eqn_hybrid}
\end{align}
As for the $V$- and $W$-components, the traveling pulse solution, $\tilde{V}(z)$ and $\tilde{W}(z)$ with the comoving frame $z = c(x - \eps^2 c t)$, satisfies (\ref{comovingV}) 
\begin{eqnarray*}
- c^2 \tilde{V}_{z} & = &  c^2 \tilde{V}_{zz} + \tilde{U} - \tilde{V}, 
\\  
- D^2 c^2 \tilde{W}_{z} & = & D^2 c^2 \tilde{W}_{zz} + \tilde{U} - \tilde{W},
\end{eqnarray*}
where we again used $(\hat{\tau}, \hat{\theta}) = (\eps^2 \tau, \eps^2 \theta) = (1, D^2)$. 
Then, we solve the linear ODEs 
\begin{align}
\label{profileV}
\tilde{V}(z)  =  \left\{ 
\begin{aligned}
& \frac{2 \rho_{v-}}{ \rho_{v+} - \rho_{v-}} \left( e^{-\rho_{v+} l_2} - e^{-\rho_{v+} l_1}  \right) e^{\rho_{v+} z} - 1 \, , && z \le l_1, \\
&\frac{2 \rho_{v-}}{ \rho_{v+} - \rho_{v-}} e^{\rho_{v+} (z-l_2)} - \frac{2 \rho_{v+}}{ \rho_{v+} - \rho_{v-}}  e^{\rho_{v-} (z -l_1)}  - 1 \, , && l_1 < z \le l_2, \\
&\frac{2 \rho_{v+}}{ \rho_{v+} - \rho_{v-}} \left( e^{-\rho_{v-} l_2} - e^{-\rho_{v-} l_1}  \right) e^{\rho_{v-} z} - 1 \, ,  && z > l_2,  \\
\end{aligned}
\right. 
\end{align}
where $\displaystyle{ \rho_{v\pm} :=  (-1 \pm \phi_v)/2 }$. Replacing $\rho_{v\pm}$
by $\displaystyle{ \rho_{w\pm} :=  (-1 \pm \phi_w)/2 }$ in (\ref{profileV}) gives the profile $\tilde{W}(x)$.
Observe that, upon setting $l_1 = -z^*$ and $l_2 = z^*$ in the above, the two ODEs in (\ref{eqn_hybrid})
with $\dot{l}_2 = \dot{l}_1 = c$ coincide with the existence criterion 
(\ref{existence_criterion_one})/(\ref{existence_criterion}) obtained from the action functional approach. 

Next, we investigate the eigenvalue problem for (\ref{eqn_hybrid}) as $L \Psi = \lambda \Psi$ with 
$\Psi(z) := (\psi_1, \psi_2, p(z), q(z))^T $ given by
\begin{align}
\label{eigen_hybrid}
\left \{ 
\begin{aligned}
 \frac{3 }{\sqrt{2}} ( \alpha ( c \tilde{V}_z(l_1) \psi_1 + p(l_1)) 
+ \beta (c \tilde{W}_z(l_1) \psi_1 + q(l_1) ) )  & = \hat{\lambda} \psi_1 \, , \\
 - \frac{3}{\sqrt{2}} (  \alpha (c \tilde{V}_z(l_2) \psi_2 + p(l_2))
+ \beta ( c \tilde{W}_z(l_2) \psi_2 + q(l_2) ) )  & = \hat{\lambda} \psi_2 \, ,  \\
 \hat{\tau} c^2 p_{zz} +  \hat{\tau} c^2 p_z - p + \psi_2 \delta(z - l_2) - \psi_1 \delta(z - l_1) & = \hat{\tau} \hat{\lambda} p \, , \\
  \hat{\theta} c^2 q_{zz} +  \hat{\theta} c^2 q_z - q + \psi_2 \delta(z - l_2) - \psi_1 \delta(z - l_1) & = \hat{\theta} \hat{\lambda} q \, , 
\end{aligned}
\right.
\end{align}
where we rescaled $\lambda = \eps^2 \hat{\lambda}$ and with $\delta$ the standard Kronecker delta-function.
Clearly, $(\psi_1, \psi_2, p, q) =  (1,1,\tilde{V}_z(z),\tilde{W}_z(z))$ is a solution of (\ref{eigen_hybrid}) associated 
to the translation free zero $\hat{\lambda} = 0$.
We solve the last two equations for $p(z)$  and $q(z)$ with the suitable conditions 
\begin{align*}
 p_z(l_1 - 0) - p_z(l_1 + 0) = - \frac{2 \psi_1}{c \hat{\tau}}  \, , &  \qquad
 p_z(l_2 - 0) - p_z(l_2 + 0) = \frac{2 \psi_2}{c \hat{\tau}} ,  \\
q_z(l_1 - 0) - q_z(l_1 + 0) = - \frac{2 \psi_1}{c \hat{\theta}}  \, , & \qquad  q_z(l_2 - 0) - q_z(l_2 + 0) = \frac{2 \psi_2}{c \hat{\theta}} . 
\end{align*}
Then, we obtain 
\begin{align*}
  p(l_1) = -\frac{2}{c \hat{\tau}(\kappa_{v+} - \kappa_{v-})} \left( \psi_1 - \psi_2 e^{-\kappa_{v+} h}  \right),  & \,  p(l_2) = \frac{2}{c \hat{\tau}(\kappa_{v+} - \kappa_{v-})} \left( \psi_2 - \psi_1 e^{\kappa_{v-} h}  \right),   \\
 q(l_1) = -\frac{2}{c \hat{\theta}(\kappa_{w+} - \kappa_{w-})} \left( \psi_1 - \psi_2 e^{-\kappa_{w+} h}  \right),  &\,   q(l_2) = \frac{2}{c \hat{\theta}(\kappa_{w+} - \kappa_{w-})} \left( \psi_2 - \psi_1 e^{\kappa_{w-} h}  \right),
\end{align*}
where $\displaystyle{ \kappa_{v\pm} = \kappa_{v\pm}(\hat{\tau}):= \frac{1}{2} \left( -1 \pm \sqrt{1 + \frac{4}{c^2 \hat{\tau}} ( 1 + \hat{\tau} \hat{\lambda} )} \right) }$, $\displaystyle{ \kappa_{w \pm} := \kappa_{v\pm}}(\hat{\theta})$,
and $h := l_2 - l_1$. 
We also have 
\begin{align}
\label{VVVV}
\begin{aligned}
  \tilde{V}_z(l_1) = \frac{2 \rho_{v+} \rho_{v-}}{\rho_{v+} - \rho_{v-}} \left( e^{-\rho_{v+} h} - 1 \right) \, ,  & \quad \tilde{V}_z(l_2) = \frac{2 \rho_{v+} \rho_{v-}}{\rho_{v+} - \rho_{v-}} \left(1- e^{\rho_{v-} h} \right), 
\end{aligned}
\end{align}
and $\tilde{W}_z(l_1)$ and $ \tilde{W}_z(l_2)$ are obtained by replacing $\rho_{v\pm}$ with $\rho_{w\pm}$ in the above expressions for $\tilde{V}_z(l_1)$ and $ \tilde{V}_z(l_2)$, respectively.

Substituting (\ref{VVVV}) into (\ref{eigen_hybrid}), and taking into account that the first two equations has a non-trivial solution 
with respect to $(\psi_1, \psi_2)$, we get the following result.
\begin{theorem} 
Let $\tau, \theta = {\cal O}(1/\varepsilon^2)$ and  $ (\hat{\tau}, \hat{\theta}) = (\varepsilon^2 \tau, \varepsilon^2 \theta)  = (1, D^2)$, and let $(\alpha, \beta, \gamma)$ be such that there exist a traveling $1$-pulse solution. The eigenvalues $\hat{\lambda}$ associated with the stability of the traveling $1$-pulse solution to 
(\ref{eqn_hybrid}) is determined by 
\begin{align}
\label{hybrid_eigen}
\begin{aligned}
0 & = 
\left(  \frac{2 c \alpha \rho_{v+} \rho_{v-}}{\rho_{v+} - \rho_{v-}} \left( e^{\rho_{v-} h} - 1 \right)
- \frac{2 \alpha}{c\hat{\tau}(\kappa_{v+} - \kappa_{v-})}+ \frac{2 c \beta \rho_{w+} \rho_{w-}}{\rho_{w+} - \rho_{w-}} \left( e^{\rho_{w-} h} - 1 \right) \right.    \\  
 & \quad
\left.
- \frac{2 \beta}{c\hat{\theta}(\kappa_{w+} - \kappa_{w-})}
- \frac{\sqrt{2} \hat{\lambda}}{3} \right)   
 \left(  \frac{2 c \alpha \rho_{v+} \rho_{v-}}{\rho_{v+} - \rho_{v-}} \left( e^{-\rho_{v+} h} - 1 \right)
- \frac{2 \alpha }{c \hat{\tau}(\kappa_{v+} - \kappa_{v-})} \right.  \\  
& \quad  \left.
+  \frac{2 c \beta \rho_{w+} \rho_{w-}}{\rho_{w+} - \rho_{w-}} \left( e^{-\rho_{w+} h} - 1 \right)
- \frac{2 \beta }{c \hat{\theta}(\kappa_{w+} - \kappa_{w-})} 
- \frac{\sqrt{2} \hat{\lambda}}{3} \right) \\
&\quad- \frac{4}{c^2} \left( \frac{\alpha e^{\kappa_{v-} h}}{\hat{\tau} (\kappa_{v+} - \kappa_{v-})}
+  \frac{\beta e^{\kappa_{w-} h}}{\hat{\theta} (\kappa_{w+} - \kappa_{w-})} \right) 
\left( \frac{\alpha e^{-\kappa_{v+} h}}{\hat{\tau} (\kappa_{v+} - \kappa_{v-})}
+  \frac{\beta e^{-\kappa_{w+} h}}{\hat{\theta} (\kappa_{w+} - \kappa_{w-})} \right). 
\end{aligned}
\end{align}
\end{theorem}
We denote the function defined in the right-hand side of the above equation 
by $\mathcal{G}(c, \hat{\theta}, \hat{\lambda})$. It is easily found that $\mathcal{G}(c, \hat{\theta}, 0) = 0$ holds for $\hat{\lambda}=0$, corresponding to the translation invariance. 

Next, we investigate the fate of the other real root of (\ref{hybrid_eigen}), that is, 
the zero eigenvalue corresponds to the saddle-node bifurcation as indicated in Fig. \ref{fig02}.  
At the turning point of the branch of traveling pulse solutions, the equation $\mathcal{G}(c, \hat{\theta}, \tilde{\lambda}) =0$ 
has double zero root. Then, both of $\displaystyle{ \frac{\partial \mathcal{G}}{\partial \tilde{\lambda}}  =0}$ 
and $\displaystyle{\mathcal{G}=0}$ holds at $\tilde{\lambda} = 0$. 
Expanding $\mathcal{G}(c, \hat{\theta}, \tilde{\lambda})$ with respect to $\tilde{\lambda}$,  we obtain  
\begin{align}
\label{hybrid_expand}
\begin{aligned}
\mathcal{G}(c, \hat{\theta}, \tilde{\lambda}) & =  \left(\frac{\alpha  e^{-\phi_v z^*}}{c \hat{\tau} \phi_v} 
+ \frac{\beta e^{-\phi_w z^*}}{c \hat{\theta} \phi_w} \right)
\left( - \frac{8 \alpha}{c^3 \hat{\tau} \phi_v^3} \left(  e^{z^*} + e^{-z^*} - 2 e^{-\phi_v z^*} \right. 
\right.  \\ &\quad
\left. - 2 \phi_v  z^* e^{- \phi_v z^*} \right) 
- \frac{8 \beta}{c^3 \hat{\theta} \phi_w^3} \left(  e^{z^*} + e^{-z^*} - 2 e^{-\phi_w z^*} - 2 \phi_w  z^* e^{- \phi_w z^*} \right) 
\\
& \quad
 \left. + \frac{2 \sqrt{2}}{3} \left(e^{z^*} + e^{-z^*}  \right)  \right)  \tilde{\lambda} 
+ {\mathcal O}(\tilde{\lambda}^2), 
\end{aligned}
\end{align}
where we replace the pulse width $h$ by $2 z^*$. 
Finally, we find that  $\displaystyle{ \frac{\partial \mathcal{G}}{\partial \tilde{\lambda}}  = 0}$ at $\tilde{\lambda} = 0$ 
coincides with the stability criterion of $\displaystyle{ \frac{\partial}{\partial c} J_c(\bar{Z}_{p})= 0 }$ (\ref{CJzero}). 
The eigenvalue $\tilde{\lambda}$ changes its sign from plus to minus at the turning point, 
corresponding to the change of $\displaystyle{ \frac{\partial}{\partial c} J_c(\bar{Z}_{p})}$ 
from negative to positive.  

\section{Traveling $2$-pulse solutions}
\label{SS:APPL}
In this section, we apply the methodology demonstrated in \S\ref{SS:PROF}-\S\ref{SS:ACTION} to the case of the right-going traveling $2$-pulse solutions and derive the results as stated in Theorem~\ref{TH:2P}, see also Fig.~\ref{fig01_NEW}(b). Furthermore, we will deduce the necessary condition for the existence of traveling $2$-pulse solutions as stated in Lemma~\ref{L:2P} from the existence condition (\ref{existence_criterion_two}) of Theorem~\ref{TH:2P}.

\subsection{The profile and action functional of a traveling $2$-pulse solution}
\label{SS:ACTION2}
For traveling $2$-pulse solutions we have to divide the spatial domain into five slow regions $I_s^{1,3,5,7,9}$ and four fast regions $I_f^{2,4,6,8}$. In particular,
\begin{eqnarray}
\label{SLOWFAST2}
\begin{array}{rl}
&I_s^1:=(-\infty,z_1-\se]\,, \,\, I_f^2:=(z_1-\se,z_1+\se)\,,  \\
& I_s^3:=[z_1+\se, z_2-\se]\,, \,\,
I_f^4:=(z_2-\se,z_2+\se)\,, \\
& I_s^5:=[z_2+\se, z_3-\se]\,, \,\,
I_f^6:=(z_3-\se,z_3+\se)\,, \\
& I_s^7:=[z_3+\se, z_4-\se]\,, \,\,
I_f^8:=(z_4-\se,z_4+\se)\,, \\
& I_s^{9}:=[z_4+\se, \infty)\,,  
\end{array}
\end{eqnarray}  
where $z_{1,2,3,4}$ are the locations of the four interfaces of a traveling $2$-pulse solution $\bar{Z}_{2p}$, 
that is, $\bar{U}(z_i)=0$ for $i \in \{1,2,3,4\}$, see \cite{vH_EXIS, vH_ACTION1}. 
We say that the widths between interfaces are given by $2 y_1 := z_2 - z_1$, 
$2 y_2 := z_3 - z_2$ and $2 y_3 := z_4 - z_3$, and, without loss of generality, we assume that $\displaystyle{ \sum_{i=1}^4z_i = 0 }$, see Fig. \ref{fig04_DOM}. 
\begin{figure}
 \centering
 \includegraphics[width=10cm]{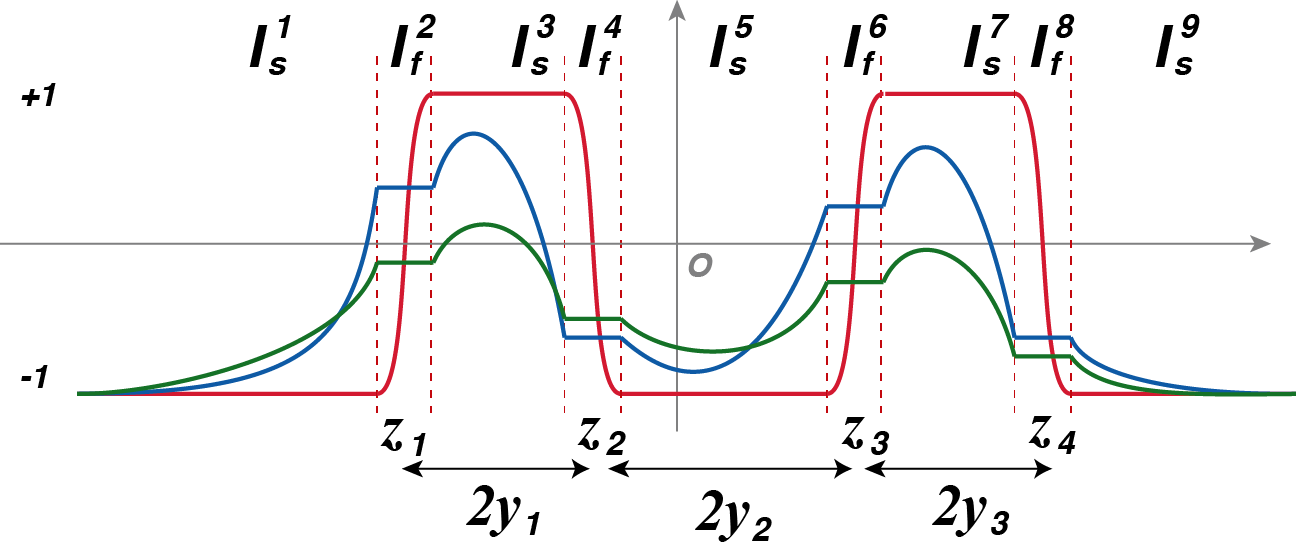}
  \caption{ 
    Schematic picture of the five slow region $I_s^{1,3,5,7,9}$ and four fast regions $I_f^{2,4,6,8}$ introduced in (\ref{SLOWFAST2}) and used in this article to study traveling $2$-pulse solutions.
}
  \label{fig04_DOM}
\end{figure}

As in \S\ref{S2}, a traveling $2$-pulse solution $\bar{Z}_{2p}$ is a solution to (\ref{comoving}) and we again set $\eps^2(\tau,\theta) = (\hat{\tau},\hat{\theta})=(1,D^2)$.
Upon using a regular expansion in $\eps$, $\bar{Z}_{2p}(z) = \bar{Z}_{0}(z) + \O(\eps)$ with $\bar{Z}_{0}=(\bar{U}_{0},\bar{V}_{0},\bar{W}_{0})$, we get that the $U$-component is to leading order given by
\begin{align*}
&\bar{U}_{0} (z)= \left\{
\begin{aligned}
  (-1)^{(i+1)/2}\,, \qquad& z\in I_s^i\,, \quad i \in \{1,3,5,7,9\}\,. \\
 (-1)^{1+i/2}\tanh{\left(\frac{z-z_{i/2}}{\sqrt2 \eps c}\right)}  \,, \qquad& z\in I_f^i\,, \quad i \in \{2,4,6,8\}\,.
\end{aligned}
\right.\, 
\end{align*}
Subsequently, we solve the linear equation for the $V$-component and this gives 
\begin{align}
\label{Vzero2}
&\bar{V}_{0}(z) = \left\{
\begin{aligned}
& \frac{ (1+\phi_v)}{\phi_v} e^{\frac{\phi_v-1}{2} z} \left(e^{-\frac{\phi_v-1}{2} z_1} - 
e^{-\frac{\phi_v-1}{2} z_2} + e^{-\frac{\phi_v-1}{2} z_3} - e^{-\frac{\phi_v-1}{2} z_4} \right) \\ &\qquad- 1 
\,, &z\in I_s^1\,,\\
&\frac{1}{\phi_v} - \frac{1+\phi_v}{\phi_v} e^{(1-\phi_v) y_1} \left( 1 - e^{(1-\phi_v) y_2} ( 1 - e^{(1-\phi_v) y_3} ) \right)  \,, 
& z\in  I_f^2\,,\\
& - \frac{1+\phi_v}{\phi_v} e^{\frac{\phi_v-1}{2} z} \left( e^{\frac{1-\phi_v}{2} z_2} - e^{\frac{1-\phi_v}{2} z_3} + e^{\frac{1-\phi_v}{2} z_4} \right)  & \\ 
&\qquad  + \frac{1-\phi_v}{\phi_v} e^{-\frac{1+\phi_v}{2} (z-z_1)} + 1 \,, & z\in I_s^3\,,\\
&\frac{1+\phi_v}{\phi_v} e^{(1-\phi_v) y_2} \left( 1 - e^{(1-\phi_v) y_3} \right)  + \frac{1-\phi_v}{\phi_v} e^{-(1+\phi_v) y_1} - \frac{1}{\phi_v}  \,, 
& z\in  I_f^4  \,, \\
& \frac{1+\phi_v}{\phi_v} e^{\frac{\phi_v-1}{2} z} \left( e^{\frac{1-\phi_v}{2} z_3} - e^{\frac{1-\phi_v}{2} z_4} \right)  & \\ 
&\qquad 
+ \frac{1-\phi_v}{\phi_v} e^{-\frac{1+\phi_v}{2} z} \left( e^{\frac{1+\phi_v}{2} z_1} - e^{\frac{1+\phi_v}{2} z_2} \right) - 1\,,  &z\in I_s^5\,,\\
&  \frac{1}{\phi_v} - \frac{1+\phi_v}{\phi_v} e^{(1-\phi_v) y_3} - \frac{1-\phi_v}{\phi_v} e^{-(1+\phi_v) y_2}  \left( 1 - e^{-(1+\phi_v) y_1} \right)   \,,  & z\in  I_f^6  \,, \\
& -\frac{1+\phi_v}{\phi_v} e^{\frac{\phi_v-1}{2} (z-z_4)} + 1  &  \\ 
&\qquad  
+ \frac{1-\phi_v}{\phi_v} e^{-\frac{1+\phi_v}{2} z} \left(  e^{\frac{1+\phi_v}{2} z_1} - e^{ \frac{1+\phi_v}{2} z_2} + e^{ \frac{1+\phi_v}{2} z_3} \right) \,,  &z\in I_s^7\,,\\
& \frac{1- \phi_v}{\phi_v} e^{-(1+\phi_v) y_3} \left( 1 - e^{-(1+\phi_v) y_2}( 1 - e^{-(1+\phi_v) y_1} )  \right)  - \frac{1}{\phi_v}   \,,  & z\in  I_f^8  \,, \\
& \frac{(1-\phi_v)}{\phi_v} e^{-\frac{1+\phi_v}{2} z} \left( e^{\frac{1+\phi_v}{2} z_1} - e^{\frac{1+\phi_v}{2} z_2} + e^{\frac{1+\phi_v}{2} z_3} - e^{\frac{1+\phi_v}{2} z_4} \right) - 1  \,, &z\in I_s^{9} \,, 
\end{aligned}
\right.
\end{align}
where $\displaystyle{\phi_v = \sqrt{1 + \frac{4}{c^2}}}$. 
The profile of $\bar{W}_{0}(z)$ is again given by replacing $\phi_v$ with $\displaystyle{\phi_w =}$ $\displaystyle{ \sqrt{1 + \frac{4}{c^2 D^2}}}$ in $\bar{V}_0(x)$ (\ref{Vzero2}).
A typical profile of a traveling $2$-pulse solutions is given in Fig. \ref{fig01_NEW}(b).
\begin{lemma}
The action functional $J_c$ (\ref{Jc}) of a traveling $2$-pulse solution $\bar{Z}_{2p}$ 
is given by
\begin{align}
\label{J2zero}
\begin{aligned}
 \frac{J_c (\bar{Z}_{2p})}{\varepsilon}  &= 
 \frac{2 \alpha}{\phi_v}  f_1(z_1, z_2, z_3, z_4; \phi_v) +  \frac{2 \beta}{\phi_w} f_1(z_1, z_2, z_3, z_4; \phi_w)  \\
& 
\qquad + 2 \ga (-e^{z_1} + e^{z_2}-e^{z_3} + e^{z_4})
+\frac{2 \sqrt{2}}{3} c (e^{z_1} + e^{z_2} + e^{z_3} + e^{z_4}) + \mathcal{O}(\se)
  \,, 
\end{aligned}
\end{align}
with 
\begin{align*}
f_1(z_1, z_2, z_3, z_4, \phi) & =  
e^{z_1} \left( - 1 + e^{(1-\phi) y_1} \left(1  - e^{(1-\phi) y_2} ( 1 - e^{(1-\phi) y_3} ) \right)  \right)  \\
& \quad +
e^{z_2} \left( - 1 + e^{(1-\phi) y_2} - e^{(1-\phi) (y_2+y_3)} + e^{-(1+\phi) y_1 }  \right)  \\
& \quad +
e^{z_3} \left( - 1 + e^{(1-\phi) y_3} - e^{-(1+\phi) (y_1+y_2)} + e^{-(1+\phi) y_2 }  \right)  \\
& \quad+
e^{z_4} \left( - 1 + e^{-(1+\phi) y_3} \left( 1 - e^{-(1+\phi) y_2} ( 1 - e^{-(1+\phi) y_1} ) \right)  \right) \,,
\end{align*}
where we recall that $2y_i=z_{i+1}-z_i$. 
\end{lemma}

\begin{proof}
This follows directly from a straightforward, but tedious, computation after splitting the indefinite integral $J_c$ (\ref{Jc}) into the nine regions   
\begin{align*}
J_c (\bar{Z}_{2p}) = \sum_{i=1}^5  \int_{ I_s^{2i-1}} dx +  \sum_{i=1}^4 \int_{ I_f^{2i}} dx.
\end{align*}
We omit the details of the computations.
\end{proof}
A traveling $2$-pulse solution will satisfy $J_c(\bar{Z}_{2p}) = 0$ and $\displaystyle{\frac{\delta J_c}{\delta u} \Phi =0}$. 
Substituting $J_c(\bar{Z}_{2p}) = 0$ from (\ref{J2zero}) into 
$\displaystyle{\frac{\partial J_c}{\partial y_1} =0}$, we arrive, to leading order, at 
\begin{align}
\label{DJ2zero1}
\begin{aligned}
\frac{1}{\varepsilon} \frac{\partial J_c}{\partial y_1}  &= 
\frac{2 \alpha }{\phi_v} f_2(z_1, z_2, z_3, z_4; \phi_v) + \frac{2 \beta }{\phi_w} f_2(z_1, z_2, z_3, z_4; \phi_w) \\
& \quad +
 \ga ( 3 e^{z_1} + e^{z_2} - e^{z_3} + e^{z_4}) +\frac{\sqrt{2}}{3} c (-3 e^{z_1} + e^{z_2} + e^{z_3} + e^{z_4}) 
 \\
& =  4 e^{z_1} \left( \frac{\alpha}{\phi_v} \left( 1 - (1+\phi_v) e^{(1-\phi_v) y_1} 
\left( 1 - e^{(1-\phi_v) y_2} ( 1 - e^{(1-\phi_v) y_3} ) \right) \right) \right.  \\
& \quad \left. + 
\frac{\beta}{\phi_w} \left( 1 - (1+\phi_w)e^{(1-\phi_w) y_1}  \left(  1 - e^{(1-\phi_w) y_2} ( 1 -  e^{(1-\phi_w) y_3} ) \right) \right) 
+ \ga - \frac{\sqrt{2}}{3} c \right) \, \\
& = \displaystyle{ 4 e^{z_1} \left( \alpha \bar{V}_0(z_1) + 
\beta \bar{W}_0(z_1) + \ga - \frac{\sqrt{2}}{3} c \right) = 0 } \,, 
\end{aligned}
\end{align}
where 
\begin{align*}
f_2(z_1, z_2, z_3, z_4; \phi) & = 
e^{z_1} \left( 3 -  (1 + 2 \phi) e^{(1-\phi) y_1} \left( 1 - e^{(1-\phi) y_2} ( 1 -  e^{(1-\phi) y_3} ) \right)  \right)  \\
& \quad  +
e^{z_2} \left( - 1 + e^{(1-\phi) y_2} - e^{(1-\phi) (y_2+y_3)} - (1 + 2 \phi)  e^{-(1+\phi) y_1 }  \right)   \\
& \quad  +
e^{z_3} \left( - 1 + e^{(1-\phi) y_3} + (1 + 2 \phi) e^{-(1+\phi) (y_1+y_2)} + e^{-(1+\phi) y_2 }  \right)   \\
& \quad + 
e^{z_4} \left( - 1 - (1 + 2 \phi) e^{-(1+\phi) (y_1+y_2+y_3)} - e^{-(1+\phi) (y_2+y_3)} \right.\\
&\quad \left.+ e^{-(1+\phi) y_3 }  \right). 
\end{align*}

We get the remaining three existence conditions from substituting $J_c(\bar{Z}_{2p}) = 0$ into $\displaystyle{\frac{\partial J_c}{\partial y_i} =0}$, where $i \in \{2,3\}$. In particular, 
from substituting $J_c(\bar{Z}_{2p}) = 0$ into $\displaystyle{\frac{\partial J_c}{\partial y_3} =0}$, we get to leading order
\begin{align}
\label{DJ2zero3}
 \frac{1}{\varepsilon} \frac{\partial J_c}{\partial y_3} 
 =  4 e^{z_4} \left( \alpha \bar{V}_0(z_4) + 
\beta \bar{W}_0(z_4) + \ga + \frac{\sqrt{2}}{3} c \right) = 0\,. 
\end{align}
Substituting (\ref{DJ2zero1}) into $\displaystyle{\frac{\partial J_c}{\partial y_2} =0}$ gives to leading order
\begin{align}
\label{DJ2zero2}
\begin{aligned}
 \frac{1}{\varepsilon} \frac{\partial J_c}{\partial y_2}&= 
 \frac{2 \alpha }{\phi_v} f_3(z_1, z_2, z_3, z_4; \phi_v) + \frac{2 \beta }{\phi_w} f_3(z_1, z_2, z_3, z_4; \phi_w)  \\
& \quad +
 2 \ga ( e^{z_1} - e^{z_2} - e^{z_3} + e^{z_4}) +\frac{2 \sqrt{2}}{3} c (- e^{z_1} - e^{z_2} + e^{z_3} + e^{z_4}) 
 \,, \\
& =  4 e^{z_2} \left( \frac{\alpha}{\phi_v} \left( -1 + (1+\phi_v) e^{(1-\phi_v) y_2} \left(  1 - e^{(1-\phi_v) y_3} \right) 
+ (1 - \phi_v) e^{-(1+\phi_v) y_1} \right) \right.  \\
&   \left. \quad+  
 \frac{\beta}{\phi_w} \left( -1 + (1+\phi_w) e^{(1-\phi_w) y_2}  \left( 1 - e^{(1-\phi_w) y_3} \right) 
+ (1 - \phi_w) e^{-(1+\phi_w) y_1} \right)  
\right.  \\
&   \left. \quad+ \ga + \frac{\sqrt{2}}{3} c \right)  \\
& =  4 e^{z_2} \left( \alpha \bar{V}_0(z_2) + 
\beta \bar{W}_0(z_2) + \ga + \frac{\sqrt{2}}{3} c \right) = 0 , 
\end{aligned}
\end{align}
where
\begin{align*}
f_3(z_1, z_2, z_3, z_4; \phi) &= 
 e^{z_1} \left( 1 -  e^{(1-\phi) y_1} + \phi e^{(1-\phi) (y_1+y_2)} - \phi e^{(1-\phi) (y_1 + y_2 + y_3) }  \right)  \\
& \quad +
e^{z_2} \left( 1 - \phi e^{(1-\phi) y_2} + \phi e^{(1-\phi) (y_2+y_3)} - e^{-(1+\phi) y_1 }  \right)    \\
& \quad +
e^{z_3} \left( - 1 + e^{(1-\phi) y_3} + \phi e^{-(1+\phi) (y_1+y_2)} - \phi e^{-(1+\phi) y_2 }  \right)   \\
& \quad +
e^{z_4} \left( - 1 - \phi e^{-(1+\phi) (y_1+y_2+y_3)} + \phi e^{-(1+\phi) (y_2+y_3)} + e^{-(1+\phi) y_3 }  \right).   
\end{align*}
Similarly, substituting (\ref{DJ2zero3}) into $\displaystyle{\frac{\partial J_c}{\partial y_2} =0}$, 
we get  
\begin{eqnarray} \label{DJ2zero4}
\frac{1}{\varepsilon} \frac{\partial J_c}{\partial y_2} &= &
\displaystyle{ 4 e^{z_3} \left( \alpha \bar{V}_0(z_3) + 
\beta \bar{W}_0(z_3) + \ga - \frac{\sqrt{2}}{3} c \right) } = 0.
\end{eqnarray}
In other words, the existence conditions for a traveling $2$-pulse solution in terms of the four undetermined variables $y_1, y_2, y_3$ and $c$ are given by (\ref{DJ2zero1})-(\ref{DJ2zero4}),
which coincides with (\ref{existence_criterion_two}).
Also, observe the difference in the sign in front of the $\displaystyle{\frac{\sqrt2}{3}c}$-term in (\ref{DJ2zero1})/(\ref{DJ2zero3}) and (\ref{DJ2zero2})/ (\ref{DJ2zero4}). This difference is due to the fact that $y_1$ and $y_3$ are related the a front ($U$ jumps from $-1$ to $+1$), while $y_2$ is related the a back ($U$ jumps from $+1$ to $-1$), see also Fig.~\ref{fig01_NEW}(b).

By solving (\ref{DJ2zero1})-(\ref{DJ2zero4})/(\ref{existence_criterion_two}) for $y_1, y_2, y_3$ and $c$, we get the solution branches with respect to $\gamma$ as shown in Fig. \ref{fig05}. 
At $\gamma \approx 0.873$, the traveling $2$-pulse solutions are emanated from 
the standing $2$-pulse solutions (as studied in \cite{vH_EXIS, vH_ACTION1, vH_STAB}) in a subcritical manner. 
\begin{figure}[t!]
 \centering
 \includegraphics[width=13cm]{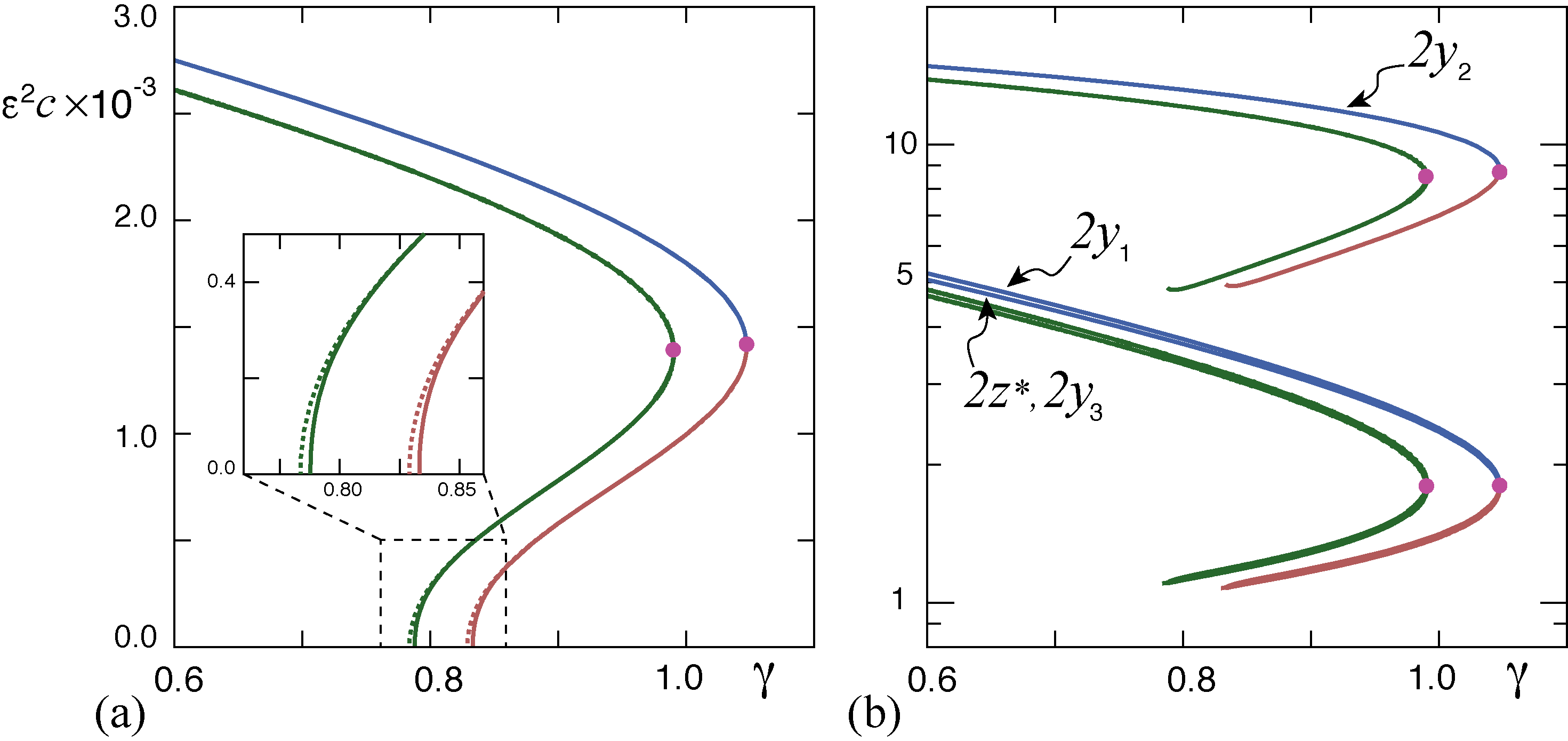}
  \caption{ The bifurcation diagram of a traveling $2$-pulse and $1$-pulse solution with respect to $\gamma$. 
  The system parameters are set to $(\al, \be, D^2) = (4, -1, 3)$. 
     The vertical axis in (a) is the propagation velocity $c$ 
     obtained by solving the existence conditions of (\ref{existence_criterion_two})
     and subsequently rescaled to the original scale of $\eps^2 c$ to compare with the numerical results. The numerical results are indicated by green curve and they are obtained by numerical continuation of the original PDE (\ref{FHNd}) with $\eps = 0.025$, see Appendix~\ref{SS:NUM} for more details on the numerical continuation. Blue and red curves indicate the solution branches obtained analytically with $\displaystyle{ \frac{\partial}{\partial c} J_c(\bar{Z}_{2p}) >0}$ and  $<0$, respectively. The inset shows the magnified figure around the onset of the solution branches. Solid and dotted curves indicate the branch curves for traveling $2$-pulse and $1$-pulse solutions. The differences between them appear only around the locations of their subcritical bifurcation points from the stationary pulse solutions.  
     The solid disks indicate the locations of turning points of solution branches. 
    The vertical axis in (b) are in log scale and depict the pulse widths $2 y_1, 2 y_2$ and $2 y_3$ for the traveling $2$-pulse solution, 
    and $2 z^*$ for the traveling $1$-pulse solutions, respectively. Again, the blue and red curves are
    obtained from (\ref{existence_criterion_two}) (and rescaled by $2 y_1/c, 2 y_2/c, 2 y_3/c$ and $2 x^*/ c$, respectively) while the green curves are obtained 
    by numerical continuation of the PDE (\ref{FHNd}).
    The curve $2 z^*$ for the traveling $1$-pulse solution width almost coincides with that of $2 y_3$ for the traveling $2$-pulse solution. 
     }
  \label{fig05}
\end{figure}
We also observe that the solution branch of the traveling $2$-pulse solutions has a turning point.
We take the derivative of the action functional $J_c$ with respect to $c$ to detect this turning point. We get
\begin{align*}
\begin{aligned}
 \frac{1}{\varepsilon} \frac{\partial J_c}{\partial c}  & =   
- \frac{8 \alpha}{c^3 \hat{\tau} \phi_v^3} f(z_1, z_2, z_3, z_4; \phi_v) 
- \frac{8 \beta}{c^3 \hat{\theta} \phi_w^3} f(z_1, z_2, z_3, z_4; \phi_w)   \\
& \quad
+\frac{2 \sqrt{2}}{3} (e^{z_1} + e^{z_2} + e^{z_3} + e^{z_4})  \,, 
\end{aligned}
\end{align*}
which coincides with (\ref{SN2}) and where $f$ is given in (\ref{func4}). 
By solving $\displaystyle{ \frac{\partial}{\partial c} J_c(\bar{Z}_{2p})= 0}$ combined with $J_c(\bar{Z}_{2p}) = 0$ and 
$\displaystyle{\frac{\partial}{\partial y_i} J_c(\bar{Z}_{2p})=0}$, we can detect 
the turning point of solution branch curve as $\gamma \approx 1.04724$ with $(y_1, y_2, y_3, c) \approx (2.03023, 9.91354, 1.99716,$ $2.25153)$ 
for  $(\al, \be, D^2) = (4, -1, 3)$, see Figure \ref{fig05}. 
Therefore, the traveling $2$-pulse solution, as long as the remaining small eigenvalues coming from the essential spectrum still have negative real part, see \cite{vH_STAB}, \S\ref{SS:COL} and Remark~\ref{R:0}, recovers their stability via saddle-node bifurcation as the derivative $\displaystyle{ \frac{\partial}{\partial c} J_c(\bar{Z}_{2p})}$ changes its sign 
from minus to plus at the turning point of the solution branch. 
Combining all of the above now gives the results as stated in Theorem~\ref{TH:2P}.

\subsection{A necessary condition for the existence of a traveling $2$-pulse solution}
\label{SS:EXIST2}
We finish this section by looking at the positive solutions $y_1$, $y_2$ and $y_3$ of the existence condition (\ref{existence_criterion_two}) to derive the necessary condition $\al\be<0$ for the existence of traveling $2$-pulse solutions (as stated in lemma~\ref{L:2P}). 
To do so, it is insightful to first look at the special parameter choice $\hat{\theta} = \hat{\tau} = 1$, that is, $D=1$\footnote{For $D=1$, (\ref{FHNd}) effectively reduces to the $2$-component model (\ref{FHN2}).}. 
The existence condition  (\ref{existence_criterion_two}) reduces to 
\begin{align*}
(\alpha + \beta) \bar{V}_i + \gamma - \frac{\sqrt{2}}{3} c = 0 \,, \quad i=1, 3, & & (\alpha + \beta)\bar{V}_j + \gamma + \frac{\sqrt{2}}{3} c = 0 \,, \quad j=2, 4.\end{align*}
Subtracting the $\bar{V}_2$ from the $\bar{V}_1$ equation yields 
\begin{align*}
\frac{2 \sqrt{2}}{3} c & = (\alpha + \beta) (\bar{V}_1 - \bar{V}_2)  \\ 
 & =   \frac{2 (\alpha + \beta)}{\phi_v} 
 \left( 1 - g_1(y_1, \phi_v) - \frac{1 + \phi_v}{2} 
 (1 - e^{(1-\phi_v) y_1} ) e^{(1-\phi_v) y_2} (1 - e^{(1-\phi_v) y_3} ) \right) 
\,, 
\end{align*}
where $\displaystyle{ g_1(y,\phi) = e^{-\phi y} (\cosh y + \phi \sinh y) }$ is a monotonically decreasing function with $\lim_{y \to + \infty} g_1(y, \phi) = 0$ for $\phi>1$.
Similarly, subtracting the $\bar{V}_4$ from the $\bar{V}_3$ equation yields
\begin{align*}
 \frac{2 \sqrt{2}}{3} c & = (\alpha + \beta) (\bar{V}_3 - \bar{V}_4)  \\ 
 & =   \frac{2 (\alpha + \beta)}{\phi_v} 
 \bigg( 1 - g_1(y_3, \phi_v) \\& \left. \quad- \frac{1 - \phi_v}{2} 
 (1 - e^{-(1+\phi_v) y_1} ) e^{-(1+\phi_v) y_2} (1 - e^{-(1+\phi_v) y_3} ) \right) 
 .
\end{align*}
Upon equating the two previous expressions, and recalling that $\phi_v>1$, we deduce that for $\hat{\theta}=1$ we necessarily have
\begin{align*}
g_1(y_3, \phi_v) - g_1(y_1, \phi_v)  &=
\frac{1 + \phi_v}{2} 
 (1 - e^{(1-\phi_v) y_1} ) e^{(1-\phi_v) y_2} (1 - e^{(1-\phi_v) y_3} ) \\&\quad  - \frac{1 - \phi_v}{2} 
 (1 - e^{-(1+\phi_v) y_1} ) e^{-(1+\phi_v) y_2} (1 - e^{-(1+\phi_v) y_3} )>0\,,  
\end{align*}
from which it follows that $y_1 >y_3$. 
\begin{figure}
 \centering
 \includegraphics[width=10cm]{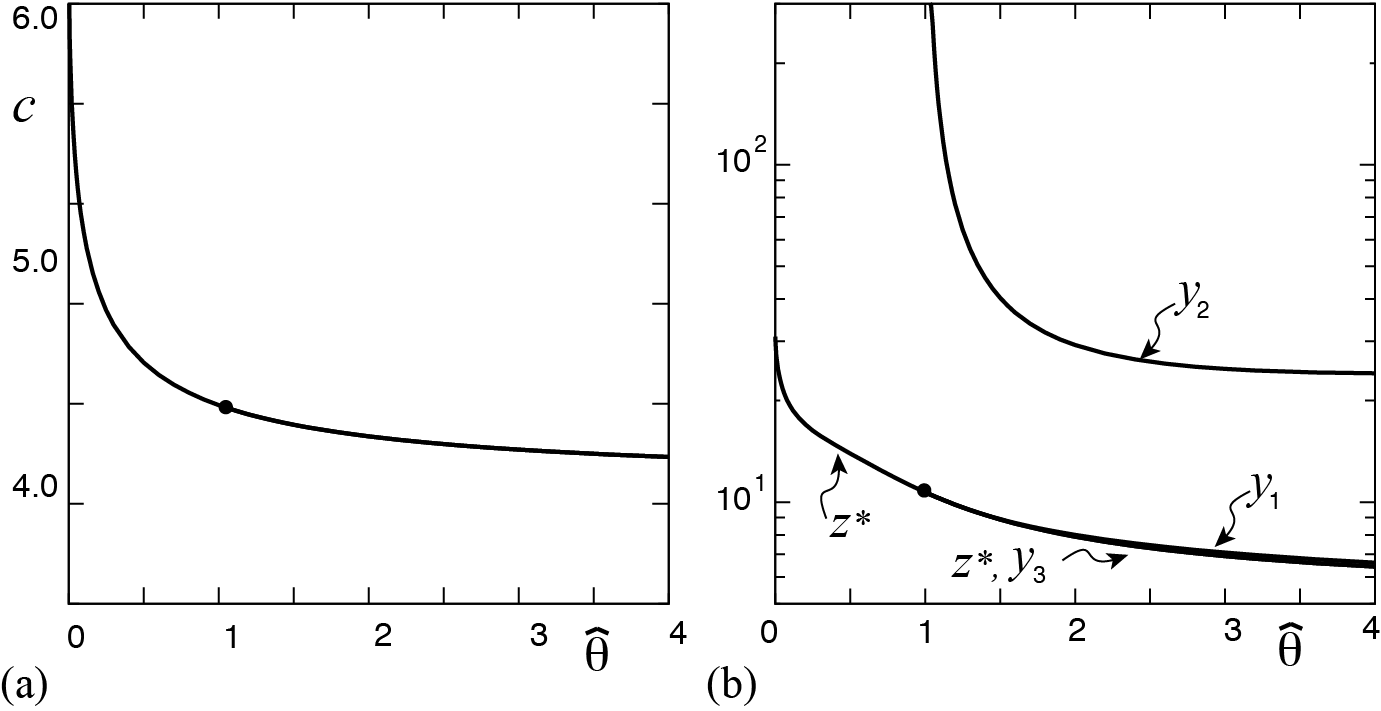}
  \caption{The bifurcation diagram of traveling $1$-pulse and $2$-pulse solutions with respect to $\hat{\theta}$, i.e., $D^2$, obtained from analytic results of Theorem~\ref{TH:1P} and Theorem~\ref{TH:2P}. The other system parameters are set to $(\al, \be, \gamma, \hat{\tau}) = (4, -1, 0.8,1)$. 
In panel (a), the vertical axis is the scaled propagating velocity $c$ and  
we observe that for $\hat{\theta}>1$ the propagating velocity $c$ of the traveling $1$-pulse and $2$-pulse solutions nearly coincide, while traveling $2$-pulse solutions do not exist for $\hat{\theta}<1$.
In panel (b), the vertical axis are the scaled pulse widths: $y_1, y_2, y_3$ for the $2$-pulse solutions and $z^*$ for the $1$-pulse solutions.  
We observe that the pulse width $z^*$ of the traveling $1$-pulse solutions nearly coincides with the pulse widths $y_1$ and $y_3$ of the traveling $2$-pulse solutions. This is not a surprise since $y_2$ is larger than $y_1$ and $y_3$, and the existence condition for the traveling $2$-pulse solution (\ref{existence_criterion_two}) reduces to leading order to twice the existence condition for a traveling $1$-pulse solution (\ref{existence_criterion_one}) -- once for $y_1$ and once for $y_3$ -- upon expanding (\ref{existence_criterion_two}) for large $y_2$.
The solid disks indicate the point $\hat{\theta}=1$ 
  where the solution branch of the $2$-pulse solutions merges into that for the $1$-pulse solutions. 
  }
  \label{fig06}
\end{figure}

Similarly, we also get  
\begin{align*}
- 2 \gamma & = (\alpha + \beta) (\bar{V}_1 + \bar{V}_2)  \\ 
 & = \frac{2 (\alpha + \beta)}{\phi_v} 
 \left( \frac{1 + \phi_v}{2} 
 (1 + e^{(1-\phi_v) y_1} ) e^{(1-\phi_v) y_2} (1 - e^{(1-\phi_v) y_3} )
 - g_2(y_1, \phi_v)  \right)
 \end{align*}
 and
 \begin{align*}
- 2 \gamma & = (\alpha + \beta) (\bar{V}_3 + \bar{V}_4)  \\ 
 & = \frac{2 (\alpha + \beta)}{\phi_v} 
 \left( - \frac{1 - \phi_v}{2} 
 (1 + e^{-(1+\phi_v) y_3} ) e^{-(1+\phi_v) y_2} (1 - e^{-(1+\phi_v) y_1} ) 
 - g_2(y_3, \phi_v) \right) , 
\end{align*}
where $\displaystyle{ g_2(y,\phi) = e^{-\phi y} (\sinh y + \phi \cosh y) }$ 
is a monotonically decreasing function with 
$\lim_{y \to + \infty} g_2(y, \phi) = 0$ for $\phi>1$. 
Subtracting the above two expressions gives
$
0 = (\al+\be)(\bar{V}_1 + \bar{V}_2-\bar{V}_3 - \bar{V}_4)
$, which implies, after some algebra, that
\begin{align*}
g_1(y_2, \phi_v) - g_1(y_1+y_2+y_3, \phi_v) &+ g_2(y_1 + y_2, \phi_v) - g_2(y_1, \phi_v)\\
&+ g_2(y_3, \phi_v)  - g_2(y_2 + y_3, \phi_v) = 0\,.
\end{align*} 
We know that $g_1(y_2, \phi_v) - g_1(y_1+y_2+y_3, \phi_v)>0$, since $g_1$ is a monotonically decreasing function and $y_i>0$ by construction. Thus, we necessarily have 
\begin{align} \label{CON}
g_2(y_1, \phi_v)  - g_2(y_1 + y_2, \phi_v) > g_2(y_3, \phi_v) - g_2(y_2 + y_3, \phi_v).
\end{align}
 The fact that $\displaystyle{ \frac{\partial^2 g_2}{\partial y^2}  = e^{-\phi y} (\phi^2-1) (\phi \cosh y  - \sinh y)> 0 }$ 
now yields that we need that $y_1 < y_3$ for (\ref{CON}) to hold. This contradicts the previous observation that $y_1 >y_3$. 
In other words, (\ref{existence_criterion_two}) has no solution for $\hat{\theta} = 1$.
Figure~\ref{fig06} shows the bifurcation diagrams and behavior of the solution branches for traveling $1$-pulse and $2$-pulse solutions upon changing $\hat{\theta}$. 
We observe that the traveling $2$-pulse solutions disappear at $\hat{\theta} = 1$ and merge into 
the solution branch of the traveling $1$-pulse solutions. 
We also observe that 
the distance $y_2$ between the two pulses of a traveling $2$-pulse solution diverges
as $\hat{\theta} \to 1^+$, see the $y_2$-branch in panel (b) of Figure~\ref{fig06}. That is, 
a traveling $2$-pulse solution splits into two traveling $1$-pulse solutions as 
$\hat{\theta} \to 1^+$.

Finally, 
we are interested in 
parameter combinations $(\alpha, \beta, \gamma,D)$ 
such that (\ref{existence_criterion_two}) has positive finite solutions 
$y_1, y_2$ and $y_3$ for $D^2 = \hat{\theta} \neq \hat{\tau} = 1$. 
By adding and subtracting, (\ref{existence_criterion_two}) can be transformed 
into \begin{align}
\label{E2P}\alpha (\bar{V}_1 +\bar{V}_2 - \bar{V}_3 - \bar{V_4} )
+  \beta (\bar{W}_1 +\bar{W}_2 - \bar{W}_3 - \bar{W_4} ) = 0\,,\end{align} with 
\begin{align}
\label{2Pab}
\begin{aligned}
\bar{V}_1 +\bar{V}_2 - \bar{V}_3 - \bar{V_4} 
& =  \frac{2}{\phi_v}
\left( g_1(y_2, \phi_v) - g_1(y_1+y_2+y_3, \phi_v) \right.  \\
& \quad \left. + g_2(y_1 + y_2, \phi_v) - g_2(y_1, \phi_v) - g_2(y_2 + y_3, \phi_v) 
+ g_2(y_3, \phi_v) \right), 
\end{aligned}
\end{align} 
and $\bar{W}_1 +\bar{W}_2 - \bar{W}_3 - \bar{W_4}$ is obtained from (\ref{2Pab}) by replacing $\phi_v$ by $\phi_w$.
 These expressions are positive for $y_1 > y_3$ and negative for $y_1 < y_3$. Thus, (\ref{E2P}) has no solutions if $\alpha \beta > 0$ and 
 we can thus conclude that a necessary condition for the existence of 
traveling $2$-pulse solutions is $\alpha \beta < 0$. 
In other words, for $\alpha \beta > 0$ 
the only potential solutions of (\ref{existence_criterion_two}) are $y_2$ is infinite and 
$y_1 = y_3$, that is, there only exist 
traveling $1$-pulse solutions.
This completes the proof of Lemma~\ref{L:2P}.

\section{Concluding remarks}
\label{SS:CONC}
In this article, we used geometric singular perturbation techniques and an action functional to show that a singularly perturbed three-component FitzHugh--Nagumo model 
supports traveling $1$-pulse and $2$-pulse solutions, see Fig.~\ref{fig01_NEW}.  
In particular, and as stated in detail in Theorem~\ref{TH:1P} and Theorem~\ref{TH:2P}, we derived explicit existence conditions as the combination of the roots of the action functional $J_c(\bar{Z}_{p,2p}) = 0$ and roots of its derivate $\displaystyle{ 
J'_c(\bar{Z}_{p,2p}) = 0 }$, where $'$ is the derivate with respect to the undermined variables of the scaled pulse width, $z^*$ for a traveling $1$-pulse solution and $y_1, y_2, y_3$ for a traveling $2$-pulse solution, and the propagating velocity $c$. Moreover, we derived the condition for a saddle-node bifurcation as $\displaystyle{\frac{\partial}{\partial c} J_c(\bar{Z}_{p,2p}) = 0 }$, see (\ref{SN}) and (\ref{SN2}), and this derivative changes from positive to negative at the turning point. This indicates that the lower branch of traveling pulse solutions is unstable, while the upper branch is potentially stable. Upon studying the existence condition (\ref{existence_criterion_two}) of Theorem~\ref{TH:2P}, we also determined a necessary condition for the existence of traveling $2$-pulse solutions, see Lemma~\ref{L:2P}.

Following this approach, we can consider the traveling $N$-pulse solutions, $\bar{Z}_{Np}$, which goes asymptotically to $(U_b, U_b, U_b)$ as $x \to \pm \infty$, that consists of $2N$ interfaces of fronts or backs of a $N$-pulse, $z_i \, (i = 1, 2, \cdots, 2N)$ with $z_i < z_{i+1}$. The leading order of the action functional can be computed explicitly after some straightforward computations (which we present without proof).
\begin{lemma}
The action functional $J_c$ (\ref{Jc}) of a traveling $N$-pulse solution $\bar{Z}_{Np}$ 
is given by
\begin{align}
\label{JNzero}
\begin{aligned}
 \frac{J_c (\bar{Z}_{Np})}{\varepsilon}  &= 
 \frac{2 \alpha}{\phi_v}  h(z_1, z_2, \cdots, z_{2N}; \phi_v) +  \frac{2 \beta}{\phi_w} h(z_1, z_2, \cdots, z_{2N}; \phi_w)  \\
& 
\qquad + 2 \ga \left( \sum_{i=1}^{2N} (-1)^i e^{z_i} \right)
+\frac{2 \sqrt{2}}{3} c \left( \sum_{i=1}^{2N} e^{z_i} \right) + \mathcal{O}(\se)
  \,, 
\end{aligned}
\end{align}
with 
\begin{align*}
& h(z_1, z_2, \cdots, z_{2N}, \phi) \\
& =  
e^{z_1} \left( - 1 - \sum_{i=1}^{2N-1} (-1)^i e^{(1-\phi) \sum \limits_{k=1}^i y_k} 
\right)  \\
& \quad +
\sum_{i=2}^{2N-1}
e^{z_i} \left( - 1 
 + \sum_{j=i}^{2N-1} (-1)^{i+j} e^{(1-\phi) \sum \limits_{k=1}^j y_k}
 - \sum_{j=1}^{i-1} (-1)^j e^{-(1+\phi) \sum \limits _{k=1}^i y_{i-k}} 
\right)  \\
& \quad+
e^{z_{2N}} \left( - 1 + \sum \limits_{i=1}^{2N-1} (-1)^i e^{-(1+\phi) \sum \limits_{k=1}^i y_{2N-k}} 
\right) \,,
\end{align*}
where we recall that $2y_i=z_{i+1}-z_i$ and $\displaystyle{\sum_{i=1}^{2N} z_i = 0}$. 
\end{lemma}
As the combination of the roots of $J(\bar{Z}_{Np})=0$ and its 
derivative 
\begin{align*}
\displaystyle{ \left \{ \frac{\partial J(\bar{Z}_{Np})}{\partial y_i} = 0 
\right \}_{i=1}^{2N-1} }, 
\end{align*}
we can derive the existence conditions of traveling $N$-pulse solutions $\bar{Z}_{Np}$. The computations will be straightforward, but extremely tedious, and we decided not to pursue this direction. 

\subsection{Collision dynamics and Hopf instabilities near turning points} \label{SS:COL}
We end this article by discussing some interesting results of numerical simulations of (\ref{FHNd}). 
Figure~\ref{fig07} shows the numerical simulations of interacting counter-propagating $1$-pulse and $2$-pulse solutions for $(\alpha, \beta, \gamma, D^2) = (4, -1, 0.8, 3)$. Note that from Theorem~\ref{TH:1P} and Theorem~\ref{TH:2P} it follows that for this parameter set traveling $1$-pulse and $2$-pulse solutions coexist and  
we take these counter-propagating pulse solutions as the initial conditions. 
As shown in Fig.~\ref{fig05}, the traveling pulse solutions 
we are dealing with are the fast type, that is, they emanate through a subcritical bifurcation 
from the stationary solutions (in contrast, the slow type are emanated through a supercritical bifurcation). These pulse solutions appear to be unstable, 
and then recover their stabilities after turning around the saddle-node points. 
In (a) two counter-propagating $1$-pulse solutions 
collide at the center part of the domain, and then they disappear and settle into the background uniform state. We observe the same phenomena for  two counter-propagating $2$-pulse solutions in (b).
In (c) we show the collision between a left-going $1$-pulse solution and a right-going $2$-pulse solution. The $1$-pulse solution and the first peak of the $2$-pulse solution annihilate after their collision and only the second peak of $2$-pulse solution survives and turns into the right-going $1$-pulse solution. \begin{figure}
 \centering
 \includegraphics[width=10cm]{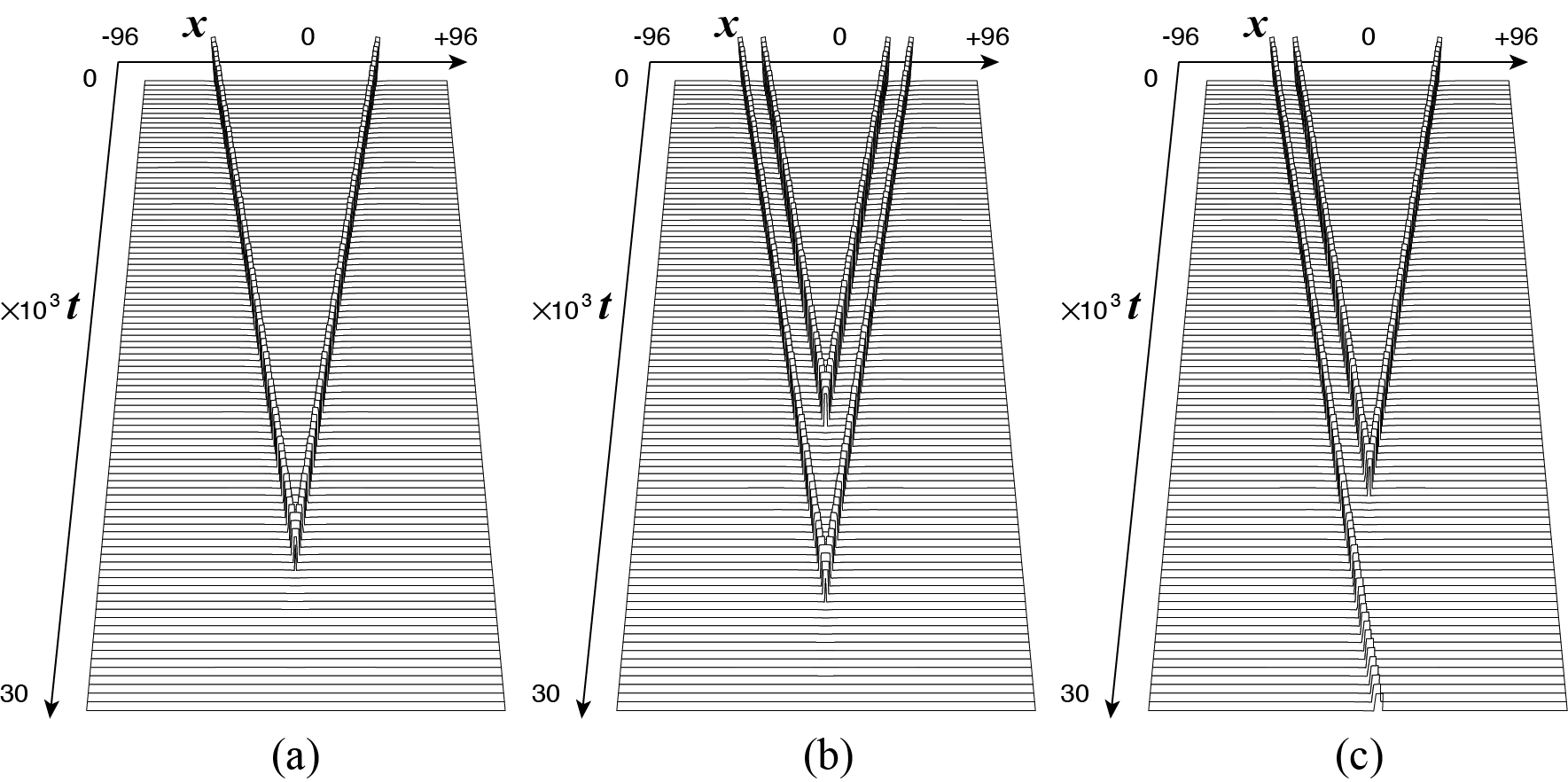}
  \caption{ Numerical simulations of (\ref{FHNd}) 
  showcasing the collision dynamics of traveling $1$-pulse and $2$-pulse solutions. (a) Collision 
  between two counter-propagating $1$-pulse solutions, (b) collision between two counter-propagating $2$-pulse solutions and  
  (c) collision between a counter-propagating $1$-pulse and $2$-pulse solution. The parameters are set to 
  $(\alpha, \beta, \gamma, D^2, \varepsilon) = (4, -1, 0.8, 3, 0.025)$ and $(\hat{\tau},\hat{\theta})=\eps^2(1,D^2)$.    }
  \label{fig07}
\end{figure}
\begin{figure}[ht!]
 \centering
 \includegraphics[width=10cm]{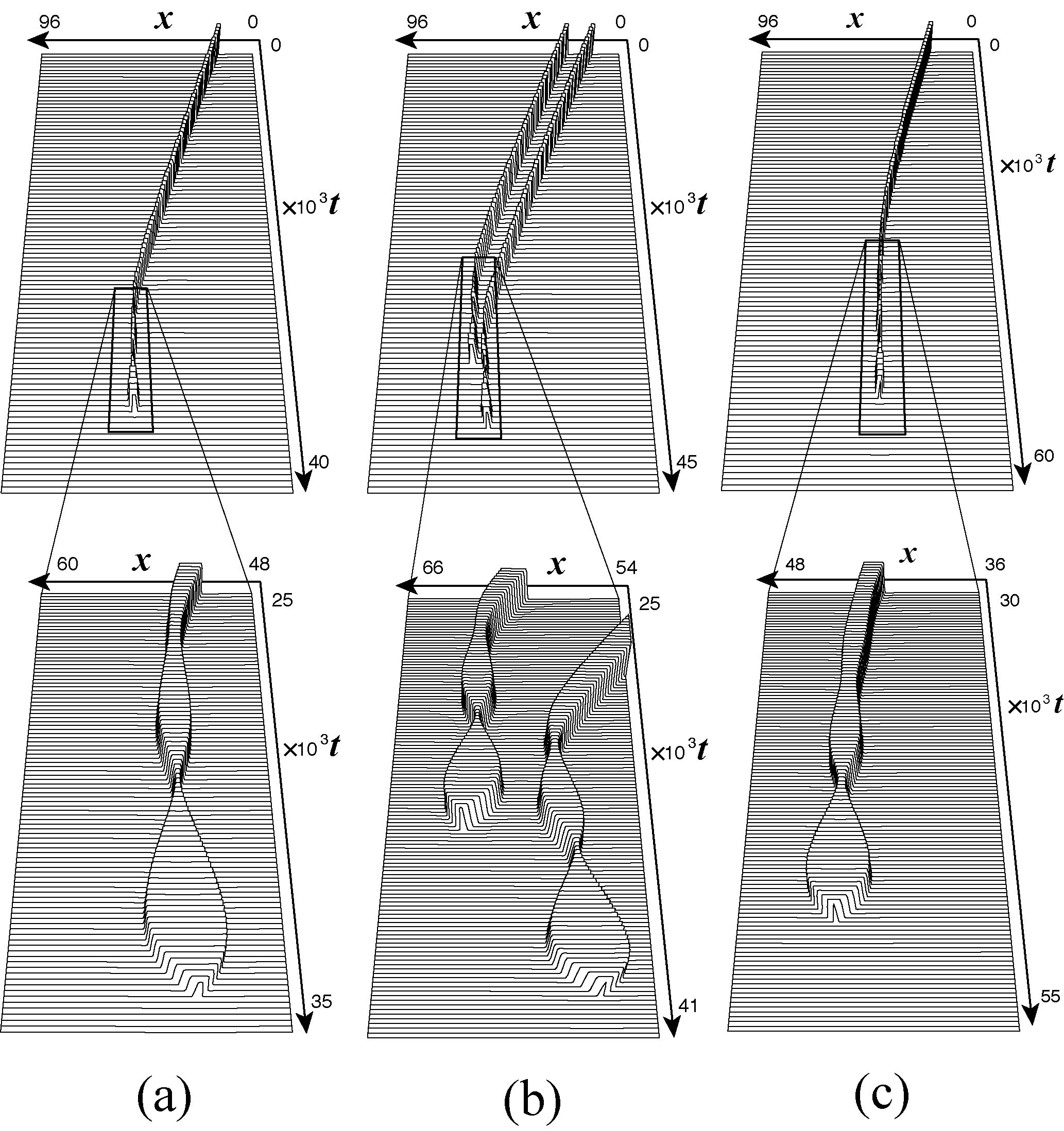}
  \caption{ Numerical simulations of 
  the oscillatory destabilization to the uniform background state near $-1$.
  In (a) and (b) the parameters are set to 
  $(\alpha, \beta, D^2, \eps) = (4, -1, 3, 0.025)$ and we observe that the $1$-pulse and $2$-pulse solutions disappear for $\gamma = 0.913$, which is near the (numerically computed) turning point $\gamma_{SN}^{num} \approx 0.990$. 
  Just before the annihilation, we observe the oscillatory behavior of the pulses as shown 
  in the magnified figures. 
 In (c) the parameters are set to  $(\alpha, \beta, \gamma, \eps) = (2, 1, 1, 0.02)$ 
 and we observe that the $1$-pulse solution disappears around  $D^2 = 0.680$, which is again near the (numerically computed) turning point $(D^2)_{SN}^{num} \approx 0.500$.  
  Just before the annihilation, we again observe the oscillatory behavior of the pulse as shown 
  in the magnified figure. 
  }
  \label{fig08}
\end{figure}

We observe stable traveling pulse solutions in the parameter regions associated with the upper parts of solution branches 
in Figs.~\ref{fig02} and \ref{fig05}. However, these pulse solutions lose their stability just before the turning point on the solution branches. 
Fig.~\ref{fig08} shows the spatio-temporal behavior of a $1$-pulse and $2$-pulse solution near the turning point. In particular, in (a) and (b) we look at the negative case of $\alpha \beta < 0$~\cite{vH_ACTION2} and set $(\al, \be, D^2) = (4, -1, 3)$ and $\gamma= 0.913$, such that $\gamma$ is in the neighborhood of the turning point at $\gamma_{SN}^{num} \approx 0.990$. The traveling $1$-pulse and $2$-pulse solutions travel with constant speed for a while, then they start to oscillate and finally they annihilate to the uniform backgrounds state near $-1$. 
These observations indicate that a Hopf instability occurs just before the turning points of the upper branches in the bifurcation diagram of Fig.~\ref{fig06} , and the traveling pulse solutions become unstable. 
In Fig.~\ref{fig08}(c) we observe similar behavior of oscillatory destabilization to the uniform background state for the positive case where $\alpha \beta > 0$. 
In particular, we set $(\al, \be, \gamma) = (2, 1, 1)$ and $D^2 = 0.680$, such that $D^2$ is in the neighborhood of the turning point at $(D^2)_{SN}^{num} \approx 0.500$. See also the bifurcation diagram of Fig.~\ref{fig02}.  By increasing $D^2$, traveling pulse solutions appear from the stable standing pulse solutions in a subcritical manner. 

The action functional approach, demonstrated in \S\ref{SS:ACTION} and \S\ref{SS:APPL}, 
did not cover the stability analysis related to the complex eigenvalues emerging from the essential spectrum upon increasing $\tau$ and/or $\theta$, see Remark~\ref{R:0}, i.e., we cannot use the action functional approach to unravel the Hopf bifurcation. %
Here, we shortly discuss how the Hopf bifurcation can also be discovered from the singular limit analysis
in \S\ref{SS:HYBD}. At a Hopf bifurcation we have a purely imaginary eigenvalue. Therefore, we set $\hat{\lambda} = i \Omega$ in (\ref{hybrid_eigen}). Moreover,
we set $\displaystyle{ \kappa_{v \pm} = \frac{-1 \pm (p_v + i q_v) }{2} }$ and $\displaystyle{ \kappa_{w \pm} = \frac{-1 \pm (p_w + i q_w) }{2} }$ with $p_v, q_v, p_w, q_w \in \mathbb{R}$ to obtain a system of four equations 
\begin{align*}
c^2 \hat{\tau}( p_v^2 - q_v^2) = c^2 \hat{\tau} + 4, 
& \qquad c^2 p_v q_v =2 \Omega\,,  \\
c^2 \hat{\theta}( p_w^2 - q_w^2) = c^2 \hat{\theta} + 4, 
& \qquad c^2 p_w q_w =2 \Omega\,.  
\end{align*}
Similarly, we can split (\ref{hybrid_eigen}) into the real and imaginary parts as follows
\begin{align}
\label{realimag}
\begin{aligned}
A_+ A_- - B^2 - C_+ C_- + D_+ D_- & =  0, \\
(A_+ + A_-) B + C_+ D_- + C_- D_+ & =  0, 
\end{aligned}
\end{align}
where 
\begin{align*}
\begin{aligned}
A_{\pm} & =   \frac{\alpha}{\hat{\tau} \phi_v} (1- e^{\mp 2 \rho_{v\pm} z^*})
- \frac{\alpha p_v}{\hat{\tau} (p_v^2 + q_v^2)} + \frac{\beta}{\hat{\theta} \phi_w} (1- e^{\mp 2 \rho_{w\pm} z^*})
- \frac{\beta p_w}{\hat{\theta} (p_w^2 + q_w^2)} ,  \\
B & =  \frac{\alpha q_v}{\hat{\tau} (p_v^2 + q_v^2)} +  \frac{\beta q_w}{\hat{\theta} (p_w^2 + q_w^2)} - \frac{\sqrt{2}}{6} c \Omega,  \\
C_{\pm} &= \frac{\alpha}{\hat{\tau} (p_v^2 + q_v^2)} e^{\mp (1 \pm p_v) z^*} \left( p_v \cos \left( q_v z^* \right) - q_v \sin \left( q_v z^* \right) \right)  \\
& \quad
+ \frac{\beta}{\hat{\theta} (p_w^2 + q_w^2)} e^{\mp (1 \pm p_w) z^*} \left( p_w \cos \left( q_w z^* \right) - q_w \sin \left( q_w z^* \right) \right),  \\
D_{\pm} & =  \frac{\alpha}{\hat{\tau} (p_v^2 + q_v^2)} e^{\mp (1 \pm p_v) z^*} \left( p_v \sin \left( q_v z^* \right) + q_v \cos \left( q_v z^* \right) \right)  \\
& \quad
+ \frac{\beta}{\hat{\theta} (p_w^2 + q_w^2)} e^{\mp (1 \pm p_w)z^*} \left( p_w \sin \left( q_w z^* \right) + q_w \cos \left( q_w z^* \right) \right). 
\end{aligned}
\end{align*}
Upon setting the parameters to $(\alpha, \beta, \gamma) = (2,1,1)$ as in Fig. \ref{fig02}, i.e.,
for the positive case of $\alpha \beta > 0$, we solve the equations of (\ref{realimag}) and (\ref{existence_criterion_one})/(\ref{existence_criterion})
with respect to $(c, z^*, \Omega, D^2)$ (recall $(\hat{\tau},\hat{\theta})=(1,D^2)$). 
 We get 
$(c, z^*, \Omega, D^2) = (2.47084, 2.58710,$ $1.86382, 0.385431)$, 
which indicates that the traveling $1$-pulse solution 
loses their stability just before the turning point of $(D^2)_{SN} 
\approx 0.359912$. 
On the other hand, setting the parameters to $(\alpha, \beta, D^2) = (4,-1,3)$ as in Fig. \ref{fig05}, i.e.,
for the negative case of $\alpha \beta < 0$, 
and solving the equations, we get $(c, z^*, \Omega, \gamma) 
= (2.52536, 2.52380, 1.35627, 1.03814)$.
Again, the Hopf bifurcation occurs in the neighborhood 
of the turning point at $\hat{\gamma}_{SN} \approx 1.04724$. 
These calculations are consistent with the numerical observations in Fig. \ref{fig08}, 
in which the traveling $1$-pulse solutions lose their stabilities via Hopf bifurcations 
just before the turning points of solution branches.  
\section*{Acknowledgements}
The authors would like to thank the 2nd Joint Australia-Japan workshop on 
dynamical systems with applications in life science (AJwsDSALS2, Biei, Japan, July 15-17, 2018) for the opportunity to work on this project. 

\appendix
\section{Numerics}
\label{SS:NUM}
In this appendix, we present a brief description of the numerical method used for path following the solution branches, demonstrated in Figs. \ref{fig02} and \ref{fig05}. 
This method is based on the predictor-corrector method of pseudo-arclength continuation \cite{Doedel, Keller}.

We begin with 
the general form of a
traveling wave problem for a scalar reaction-diffusion equation in a comoving frame $y=x-ct$  
(so for one of the components) in a one dimensional domain 
$\Omega := [-L, L]$ with periodic boundary conditions: 
$$
0 = D U_{xx} + c U_{x} + F(U; \gamma) \, , \qquad \qquad U(-L)=U(L)\,, U_x(-L)=U_x(L)\,,
$$
where $D$ is the diffusion coefficient and $F: \mathbb{R} \times \mathbb{R} \rightarrow \mathbb{R}$ is the reaction term and $\gamma \in \mathbb{R}$ represents a continuation parameter. Note that we set $L=48$ in this article.

We spatially discretize $U$ by setting $\Delta x = 2L/n$. In other words, $U$ 
becomes $\bmU:= (U_0, U_1, \cdots, U_{n-1})^T$ where 
$U_i = U(-L +i \Delta x)$, $i=0, 1, \ldots, n-1$.
So, we get
\begin{eqnarray}
\bmG(\bmU; c, \gamma)  & := & \bmD \bmU_{xx} + c \bmU_x + \bmF(\bmU; \gamma) = \bmxero, \qquad 
U_n = U_0, \; \; U_{-1} = U_{n-1},
\label{num1}
\end{eqnarray}
where $\bmD= D\bmI$ with $\bmI$ the $n \times n$ identity matrix, $\bmF$ and $\bmG: \mathbb{R}^n \times \mathbb{R} \rightarrow \mathbb{R}^n$, and $\bmxero$ is the $n$-dimensional zero-vector. Note that we set $n=12288$ in this article.

Next, we show how a single continuation step with respect to the continuation parameter $\gamma$ is implemented. Namely, the transition from a $j$-th calculated solution set to
the next solution set of the branch.
If the Jacobian matrix $\partial \bmG (\bmU^j; c^j, \gamma^j)/
\partial \bmU $ is non-singular, the implicit function theorem assures the existence of the solution branch in the neighborhood of a $j$-th solution set.
Under the periodic boundary condition, an infinite set of traveling wave solutions occur due to translation invariance of the system. To uniquely pinpoint a solution $\bmU$, we add the following integral phase condition 
\begin{eqnarray}
P(\bmU; c, \gamma) & := & \sum_{i=0}^{n-1}
U_i \cdot (U_{i+1}(s^j) - U_{i}(s^j)) = 0. 
\label{num4}
\end{eqnarray}
Since system (\ref{num1}) coupled with (\ref{num4}) 
consists of $(n+1)$ equations for $(n+2)$ unknowns. Therefore, we append these equations with 
a quadratic scalar equation for the small distance $\Delta s := s - s^j$ between two consecutive solution sets
\begin{eqnarray}
\sum_{i=0}^{n-1} (U_i - U_i(s^j))^2 \Delta x
+ (c - c(s^j))^2
+ (\gamma - \gamma(s^j))^2 - (\Delta s)^2 = 0,
\label{num2}
\end{eqnarray}
where the $j$-th solution set $\bmZ^j := (\bmU(s^j), c(s^j), \gamma(s^j))^T$
are implicitly parameterized as function of the arclength parameter $s$ along the branch.
Using a Taylor series expansion $\displaystyle{ \bmZ - \bmZ^j = \frac{d \bmZ^{j}}{d s} \Delta s + {\mathcal{O}}((\Delta s)^2) }$, we replace (\ref{num2}) by the following linear form with respect to the increments
\begin{align}
\label{num3}
\begin{aligned}
N(\bmU; c, \gamma) & =  \sum_{i=0}^{n-1}
\frac{d U_i^{j}}{d s} (U_i - U_i(s^j)) \Delta x
+ \frac{d c^{j}}{d s} (c - c(s^j))  \\
& \quad + \frac{d \gamma^{j}}{d s} (\gamma - \gamma(s^j)) - \Delta s = 0,
\end{aligned}
\end{align}
where $d \bmZ^j/ ds$ is supposed to be an unit vector tangent to the solution branch curve at the current position $\bmZ^j$. 

By solving the $(n+2)$ equations of (\ref{num1}), (\ref{num4}) and (\ref{num3})
for each continuation step, a solution branch is obtained as a chain of solutions $\bmZ^j$.
The initial guess for the next solution set is obtained in the direction of $d \bmZ^j/ ds$.
We then iteratively solve the equations using Newton's method,
\begin{eqnarray*}
\displaystyle{
\frac{\partial (\bmG, P, N)(\bmZ^j)}{\partial (\bmU, c, \gamma)}
\Delta \bmZ
} & = &
- \begin{pmatrix} \bmG(\bmZ^j) \\ P(\bmZ^j)  \\ N(\bmZ^j) \end{pmatrix},
\end{eqnarray*}
where $\Delta \bmZ := \bmZ^{j+1} - \bmZ^j$.
If the step size $\Delta s$ is given small enough, the Newton's iteration
converges to the next solution set $\bmZ^{j+1}$ on the branch
in the direction perpendicular to $d \bmZ^{j}/ d s$.
After converging, we compute the new tangent vector
$d \bmZ^{j+1}/ d s$ by solving $(n+2)$ equations using the Jacobian matrix evaluated at $\bmZ^j$
\begin{eqnarray*}
\displaystyle{
\frac{\partial (\bmG, P, N)(\bmZ^{j})}{\partial (\bmU, c, \gamma)}
\frac{d \bmZ^{j+1}}{d s}
} & = &
\begin{pmatrix} \bmxero \\ 0 \\ 1 \end{pmatrix}.
\end{eqnarray*}
The new tangent vector is rescaled to satisfy
$|d \bmZ^{j+1}/ d s|^2 = 1$,
thus preserving the right direction along the branch. 

\end{document}